\documentclass[11pt]{article}

\usepackage[T1]{fontenc}
\usepackage{verbatim}
\usepackage{amsmath,amsthm}
\usepackage{xspace}
\usepackage{pifont}
\usepackage{amssymb}
\usepackage{dsfont}
\usepackage{makeidx}
\usepackage{mathrsfs}
\usepackage{xcolor, colortbl}
\usepackage{subcaption}
\usepackage[numbers]{natbib}
\usepackage{enumitem}
\usepackage{afterpage}
\usepackage{mathtools}
\usepackage{graphics}
\usepackage[export]{adjustbox}
\usepackage{tcolorbox}
\usepackage{authblk}
\usepackage[normalem]{ulem}
\usepackage{xurl} 
\usepackage[colorlinks,citecolor=blue,urlcolor=blue]{hyperref}
\mathtoolsset{showonlyrefs}
\usepackage{centernot}
\usepackage{algorithm}
\usepackage{algpseudocode}
\usepackage{tikz}
\usetikzlibrary{arrows.meta,calc,positioning,fit}

\def\eps{\varepsilon}

\def\Q{\mathbb{Q}}
\def\P{\mathbb{P}}

\DeclareMathOperator*{\argmax}{arg\,max}

\newcommand{\bd}{\begin{displaymath}}
\newcommand{\ed}{\end{displaymath}}
\newcommand{\be}{\begin{equation}}
\newcommand{\ee}{\end{equation}}
\newcommand{\bq}{\begin{eqnarray}}
\newcommand{\eq}{\end{eqnarray}}
\newcommand{\bn}{\begin{eqnarray*}}
\newcommand{\en}{\end{eqnarray*}}

\makeatletter
\renewcommand{\ALG@name}{Algorithm}
\makeatother

\newtheorem{theorem}{Theorem}[section]
\newtheorem{lemma}[theorem]{Lemma}
\newtheorem{proposition}[theorem]{Proposition}

\newtheorem{remark}[theorem]{Remark}

\newtheorem{definition}[theorem]{Definition}
\newtheorem{assumption}[theorem]{Assumption}
\numberwithin{equation}{section}
\begin{document}

\title{Quantum-Enhanced Sampling of Schrödinger Bridges }
\author[1]{Tom Lollier}
\author[2]{Eyal Neuman}
\affil[1]{École Normale Supérieure Paris-Saclay}
\affil[2]{Department of Mathematics, Imperial College London}
\renewcommand\Authands{ and }
\date{\today}
\maketitle

\begin{abstract}
We consider the dynamic Schrödinger bridge problem on a finite state space and exploit its Markov structure to decompose the problem into an endpoint coupling and a collection of conditional Markov bridges. To sample from the conditional bridges, we develop a quantum Gibbs sampler based on quantum walks on the bridge path space. Under uniformly positive transition probabilities, we derive a spectral-gap bound and provide an explicit procedure for preparing the initial state. Compared with the update bounds of the corresponding classical Gibbs sampler, our quantum walk-query complexity improves the dependence on the time horizon $T$ from quadratic to linear. For the endpoint coupling, we adapt a quantum box-constrained Newton method to compute the Schrödinger potentials, improving the matrix-scaling complexity, at fixed accuracy, from quadratic to \(N^{3/2}\) in the state-space size \(N\). Finally, we show that exponential reweighting extends both the Schrödinger bridge formulation and the Gibbs spectral-gap estimate to models with additive running costs.
\end{abstract}

\begin{description}
\item[Mathematics Subject Classification (2020):] 81P68, 60J10, 60J22, 65C40, 90C25
\item[Keywords:] Schrödinger Bridge, quantum walk, Sinkhorn's algorithm, quantum matrix scaling, Markov bridge, quantum Gibbs sampler, quantum speedup. 
\end{description}

\section{Introduction}
In 1931, Schr\"odinger asked how a large system of independently diffusing particles would most likely evolve between two observed empirical distributions when the final observation differs from the prediction of the diffusion \cite{Schrodinger1931}. This question gave rise to the Schr\"odinger bridge problem and prompted work on time reversal and reciprocal processes \cite{leonard2014survey}. In its modern formulation, the problem is to minimise relative entropy with respect to a reference path law, subject to prescribed initial and terminal distributions \cite{tang2026foundationsschrodingerbridgesgenerative}.

When a finite-entropy minimiser exists, it is unique and preserves the reference dynamics conditional on the endpoints \cite{leonard2014survey}. Computing the bridge therefore reduces to an optimisation over endpoint couplings. Sampling a trajectory then requires sampling its endpoints from the optimal coupling and the reference process conditional on those endpoints. For Brownian reference dynamics in Euclidean space, the coupling problem is an entropically regularised optimal transport problem, with quadratic optimal transport recovered in suitable small-noise limits \cite{mikami2004zero,leonard2014survey}. For diffusion processes, a stochastic control formulation identifies relative entropy, under suitable assumptions, with a quadratic cost for modifying the drift to attain the prescribed marginals \cite{daipra1991control,chen2016relation}. Applications to Monte Carlo sampling include the bridge samplers of Bernton et al.\ \cite{bernton2019schrodingerbridgesamplers}, developed for estimating expectations and normalising constants.

Generative modelling provides a further motivation for computing Schr\"odinger bridges and sampling their trajectories. Diffusion-based models generate samples by learning to reverse a process that gradually adds noise to data \cite{sohl2015deep,song2021score}. In this setting, the bridge formulation prescribes the data distribution and a simple noise distribution as the endpoints of a stochastic evolution over a fixed time interval. De Bortoli et al.\ \cite{debortoli2021diffusion} develop an iterative method that learns the dynamics in both directions, progressively correcting discrepancies with these prescribed distributions. New samples are then generated by starting from noise and following the learned dynamics towards the data distribution. Other applications include image restoration \cite{liu2023i2sb}, modelling cellular development \cite{lavenant2024trajectory}, and learning samplers for complex distributions \cite{liu2025adjoint}. Shi et al.\ \cite{shi2026schromind} also apply approximate Schr\"odinger bridges to multimodal large language models that answer questions about images. Their method adjusts internal numerical representations towards those associated with correct answers during text generation, aiming to reduce statements unsupported by the image. Such applications extend bridge methods from generating data to improving model responses; see \cite[Section~8]{tang2026foundationsschrodingerbridgesgenerative} for a broader survey of generative applications.

In mathematical finance, martingale Schr\"odinger bridges select models that match option prices while remaining close to a reference law in relative entropy. Henry-Labord\`ere \cite{henrylabordere2019martingale} applies this approach to stochastic volatility calibration. Nutz, Wiesel, and Zhao \cite{nutz2023martingale} establish a duality with exponential utility maximisation over semistatic portfolios in a two-period market, while Guyon \cite{guyon2024joint} addresses joint SPX--VIX smile calibration through a dispersion-constrained martingale Schr\"odinger problem. These formulations combine calibration requirements with martingale constraints.

Schr\"odinger bridges also provide methods for generating financial time series. Hamdouche, Henry-Labord\`ere, and Pham \cite{hamdouche2023timeseries} estimate a path-dependent drift from observed trajectories and use generated samples to train deep hedging strategies. Alouadi et al.\ \cite{alouadi2025robust} evaluate this approach on synthetic and stock-market data, examining temporal dependence and robustness. These applications require generating whole trajectories that reproduce dependence across time, as well as distributions at individual dates.

Motivated by these applications, we analyse quantum algorithms for Schr\"odinger bridges with finite-state Markov reference dynamics. Our approach combines quantum matrix scaling for computing the endpoint coupling with quantum walks for conditional-path sampling. The results below concern positive finite-state kernels under explicit access assumptions. 

To explain this decomposition, let $\mathbb{P}$ denote the reference law of a Markov trajectory $(X_t)_{t=0}^{T}$ on a finite state space $\chi$, and let $\mu_0$ and $\mu_T$ be the prescribed initial and terminal distributions. The Schr\"odinger bridge $\mathbb{Q}^\ast$ minimises the relative entropy, or Kullback--Leibler divergence,
$\operatorname{KL}(\mathbb{Q}\,\|\,\mathbb{P})$,
over path laws $\mathbb{Q}$ with these endpoint marginals. Writing
$\mathbb{P}_{0,T}(i,j)=\mathbb{P}(X_0=i,X_T=j)$
for the reference joint endpoint law, and assuming that this law and the prescribed marginals are strictly positive, the optimal endpoint coupling has the form
\begin{equation}\label{eq-dec}
    \mathbb{Q}_{0,T}^\ast(i,j)
    =\varphi^\circ(i)\mathbb{P}_{0,T}(i,j)\varphi^\ast(j),
    \qquad i,j\in\chi.
\end{equation}
Here, $\mathbb{Q}_{0,T}^\ast$ is the joint endpoint law under $\mathbb{Q}^\ast$, and
$\varphi^\circ,\varphi^\ast:\chi\to(0,\infty)$
are the Schr\"odinger potentials, which scale the rows and columns to have sums $\mu_0$ and $\mu_T$, respectively.

To sample from $\mathbb{Q}^\ast$, we first draw an initial state $i$ from $\mu_0$, followed by a terminal state $j$ with conditional probability
$\mathbb{Q}_{0,T}^\ast(i,j)/\mu_0(i)$.
Given this endpoint pair, we sample the intermediate states from the reference Markov bridge
$\mathbb{P}(\,\cdot\mid X_0=i,X_T=j)$.
This procedure follows from the factorisation
\[
    \mathbb{Q}^\ast(x)
    =\mathbb{Q}_{0,T}^\ast(x_0,x_T)\,
      \mathbb{P}(x\mid X_0=x_0,X_T=x_T),
    \qquad x=(x_0,\ldots,x_T)\in\chi^{T+1}.
\]
Thus, computing and sampling the endpoint coupling can be treated separately from conditional path sampling. We briefly review the quantum tools for these two tasks below.

\paragraph{Quantum walks and quantum Gibbs samplers.}
Classical random walks evolve probability distributions through stochastic transitions, whereas quantum walks evolve complex probability amplitudes through unitary operations. Amplitudes associated with different paths to the same outcome can interfere constructively or destructively, altering the probability of observing that outcome \cite{aharonov1993quantum}. Watrous \cite{watrous2001simulations} established an early algorithmic connection by showing how classical random walks on undirected graphs can be simulated by unitary quantum computations.

Szegedy's construction associates a quantum walk with a classical reversible Markov chain \cite{Szegedy}. In quantum sampling, the objective is to prepare a coherent encoding of the stationary distribution,
\[
    |\pi\rangle=\sum_x\sqrt{\pi(x)}\,|x\rangle,
\]
whose measurement yields a sample from $\pi$. A key feature of this construction is its spectral-gap dependence: for a finite, irreducible, lazy reversible chain with spectral gap $\delta$, the associated quantum walk has a phase gap of order $\sqrt{\delta}$ on the relevant invariant subspace. This permits approximate reflections about the stationary state using order $\delta^{-1/2}$ walk queries, up to accuracy-dependent factors. Quantum fast-forwarding provides access to finite-time Markov-chain evolution in a quantum encoding \cite{apers2019quantumfastforwardingmarkovchains} and can be formulated within the more general framework of quantum singular value transformation \cite{gilyen2019qsvt}. Such coherent access does not, by itself, provide classical trajectory samples satisfying the marginal constraints of a Schr\"odinger bridge.

Metropolis--Hastings (MH) samples from a target distribution $\pi$ by constructing a Markov chain with stationary distribution $\pi$. At each step, it proposes a new state and either accepts it or remains at the current state. The acceptance probability is chosen to make the chain reversible with respect to $\pi$ and requires only ratios of unnormalised target weights, together with the proposal probabilities \cite{metropolis1953equation,hastings1970monte}. Once the chain has mixed, its state provides an approximate sample from $\pi$. Gibbs sampling is a special case of MH: it selects a coordinate at random and draws its new value from the conditional distribution given all other coordinates. Because the proposal already follows this conditional distribution, every proposal is accepted and there is no rejection step.

In a quantum implementation, proposal and acceptance operations must act coherently, without measuring their outcomes. Lemieux et al.\ \cite{lemieux2020efficient} construct explicit circuits for Metropolis--Hastings updates in local Ising models, including both acceptance and rejection. An auxiliary coin encodes these alternatives: the accepting branch applies the proposed move, while the rejecting branch leaves the state unchanged. Their circuit implements the full quantum walk without the costly preparation of the complete transition state. Claudon et al.\ \cite{claudon2026quantumcircuitsmetropolishastingsalgorithm} provide a more general construction based on a dual Markov chain on pairs of states. This also accommodates rejection, but retains proposal information in the enlarged state space, avoiding explicit computation of the total rejection probability. Each walk step uses a constant number of calls to coherent proposal and acceptance routines. See Section \ref{sec-q-gibbs} for a detailed explanation of the method. 

\paragraph{Matrix scaling and the endpoint problem.}
Computing the endpoint coupling in \eqref{eq-dec} amounts to finding scaling factors that give the reference matrix $\mathbb{P}_{0,T}$ the prescribed row and column sums. Sinkhorn's algorithm does this by alternately adjusting the rows and columns. Its foundations lie in the work of Sinkhorn \cite{sinkhorn1964relationship,sinkhorn1967diagonal} and Sinkhorn and Knopp \cite{sinkhornKnopp1967}. Cuturi \cite{cuturi2013sinkhorn} subsequently popularised this approach for computing entropically regularised optimal transport.

In logarithmic scaling variables, the problem admits a convex formulation, and Sinkhorn iteration becomes alternating minimisation over the row and column variables. This perspective also leads to second-order algorithms, including the box-constrained Newton and interior-point methods of Cohen, M\k{a}dry, Tsipras, and Vladu \cite{cohen2017matrixscalingbalancingbox}.

Quantum algorithms have been developed for both approaches. Van Apeldoorn et al.\ \cite{vanapeldoorn2020quantumalgorithmsmatrixscaling} use amplitude estimation to accelerate Sinkhorn updates. Gribling and Nieuwboer \cite{Improved-quantum-lower-and-upper-bounds-for-matrix-scaling} combine amplitude estimation and quantum graph sparsification within a box-constrained Newton method. Their bounds depend on accuracy and the assumed quantum access to the input matrix. Section~\ref{sec:quantum-sinkhorn-alg} applies the latter method to compute an implicit representation of the endpoint coupling through its scaling factors. Drawing an endpoint pair remains an additional computational task.

We present our results on Markov bridge sampling using a quantum Gibbs sampler and on computing the Schrödinger potentials using a quantum matrix scaling method. These methods improve the state-size dependence of the endpoint computation and provide a square-root improvement in the dependence on the inverse spectral gap for the local conditional sampler.

Our analysis distinguishes the number of calls to transition and state-preparation primitives from the total computational time. Throughout this overview, \(0<\varepsilon<1\) denotes the required total variation accuracy: the probability of any event differs from its target probability by at most \(\varepsilon\).

\textbf{Markov bridge sampling.} For fixed endpoints $a$ and $b$ with $P^T(a,b)>0$, the challenge is to sample an entire trajectory from the reference Markov bridge $\mathbb{P}_{a,b}$, given by
\[
    \mathbb{P}_{a,b}(x_1,\ldots,x_{T-1})
    =\frac{\prod_{t=0}^{T-1}P(x_t,x_{t+1})}{P^T(a,b)},
    \qquad x_0=a,\quad x_T=b.
\]
Although the number of possible interior paths is $N^{T-1}$, the Markov structure permits local updates: a Gibbs sampler selects an intermediate time and resamples the state using only its two neighbours. These updates preserve the endpoints and do not require evaluating the probability of the endpoint event. The main analytical challenge is to determine how rapidly the sampler converges to the conditional distribution of complete trajectories.

Theorem~\ref{cor:upper-bound-spectral-gap} addresses this difficulty by establishing a lower bound on the spectral gap of the Gibbs sampler on path space:
\[
    \delta \geq \frac{\omega_-^4}{T^2N^3},
\]
where $\omega_-$ is the uniform positivity parameter for the reference kernel $P$, defined in \eqref{om}. The proof uses the canonical paths method. This yields a polynomial bound on the sampler's relaxation time despite the exponentially large path space and provides the quantitative input for its quantum counterpart.

Building on the quantum-walk framework of \citet{claudon2026quantumcircuitsmetropolishastingsalgorithm}, Theorem~\ref{th:complexity-gibbs-sampler} establishes a quantum sampler for $\mathbb{P}_{a,b}$. Conditional on a detectable success event with probability bounded below by a universal constant, measuring the prepared state yields a trajectory whose distribution is within total variation distance $\varepsilon$ of $\mathbb{P}_{a,b}$. The algorithm uses
\[
    O\!\left(
        \frac{TN^2}{\omega_-^{5/2}}
        \log\frac{N}{\omega_-\varepsilon}
    \right)
\]
calls to the bridge-walk primitives: the coherent Gibbs-update oracle $O_K$, its adjoint $O_K^*$, and the register-swap operator $S$, defined in Section~\ref{sec-q-gibbs}. It also uses $O(T\sqrt{N/\omega_-})$ calls to coherent row preparation for the reference kernel $P$ and its inverse.

The initial state is prepared by running the reference dynamics for
$T-1$ steps from $a$ and appending $b$. Its overlap with the target
coherent bridge state is at least $\sqrt{\omega_-/N}$. Combining this
bound with the quantum walk's square-root improvement in spectral-gap
dependence and an error bound for the measured path distribution yields
the stated sampling guarantee.

Using the same initial path distribution, we obtain a classical bound of
\[
    O\!\left(
        \frac{T^2N^3}{\omega_-^4}
        \log\frac{N}{\omega_-\varepsilon}
    \right)
\]
Gibbs updates (see Remark~\ref{rem:classical-gibbs-comparison}). Up to the displayed logarithmic factor, the quantum bound is linear in $T$, whereas the classical Gibbs bound is quadratic. These bounds concern the specified procedures and do not establish optimal classical or quantum complexities. Indeed, direct classical rejection sampling from the reference chain started at $a$ already requires at most $TN/\omega_-$ forward transitions in expectation. The contribution is therefore a bridge-specific analysis of the spectral gap, initialisation and sampling error for local Gibbs and quantum-walk methods; it does not establish an unconditional quadratic speedup in the horizon.

\paragraph{Computing the endpoint coupling.}
Theorem~\ref{th:quantum-static-bridge} addresses the complementary task
of computing the Schr\"odinger potentials. Using the quantum
box-constrained Newton method of
\citet{Improved-quantum-lower-and-upper-bounds-for-matrix-scaling},
we compute approximate Schr\"odinger potentials
$\varphi^\circ$ and $\varphi^\ast$ in \eqref{eq-dec}, which
implicitly specify an endpoint coupling. Our analysis translates the method's optimisation
error into a bound on the total variation distance from the exact
coupling, ensuring accuracy for the full joint endpoint distribution.

With probability at least $2/3$, the algorithm returns two scaling
vectors that specify an approximate endpoint coupling within total
variation distance $\varepsilon$ of the exact coupling. The quantum
running time is
\[
    \widetilde O\!\left(
        \frac{B^2N^{3/2}}{\varepsilon^2}
    \right),
\]
where $B$ bounds the magnitude of the logarithmic potentials and
$\widetilde O$ suppresses logarithmic factors.
Remark~\ref{rem:static-uniform-positivity} shows that, when the transition
probabilities and both prescribed marginals satisfy the stated uniform
positivity condition, one can take $B=O(\log(N/\omega_-))$.
For fixed $\omega_-$, the complexity therefore simplifies to
$\widetilde O(N^{3/2}/\varepsilon^2)$.

Remark~\ref{rem:static-speedup} compares this result with the
$\widetilde O(BN^2)$ bound for the classical box-constrained Newton method. At fixed accuracy and with $B$ growing at most polylogarithmically, the state-size dependence improves from $N^2$ to $N^{3/2}$. The advantage is precision-dependent: comparison of the bounds gives an improvement when
$\varepsilon^2\gg B/\sqrt N$.  

These endpoint-computation bounds assume the coherent input access
and computational model of Assumption~\ref{assump-box}; constructing
the endpoint kernel $P^T$ and its oracle representation incurs
additional costs. Once the scaling vectors have been computed,
an endpoint pair can be sampled using $O(N)$ additional entry
evaluations and arithmetic operations. As detailed in
Section~\ref{rem:static-output}, running the endpoint algorithm with
accuracy $\varepsilon/2$ ensures that the resulting sample has
total variation error at most $\varepsilon$.

Combining this procedure with the quantum conditional-path sampler
yields a complete quantum-enhanced sampling framework for the dynamic
Schr\"odinger bridge. The two quantum improvements have distinct scopes:
the endpoint routine improves the state-size dependence of matrix
scaling under Assumption~\ref{assump-box}, while the conditional
routine provides a square-root improvement in spectral-gap dependence
for the specified local Gibbs dynamics. Figure~\ref{fig:sb-algorithm-overview} summarises the endpoint and Markov bridge components of the problem.

\begin{figure}[!htb]
  \centering
  \captionsetup{font=small}
\begingroup
\definecolor{SBink}{HTML}{203040}
\definecolor{SBline}{HTML}{637484}
\definecolor{SBblue}{HTML}{245B83}
\definecolor{SBbluefill}{HTML}{EFF5FA}
\definecolor{SBteal}{HTML}{176B63}
\definecolor{SBtealfill}{HTML}{EEF7F4}
\definecolor{SBamber}{HTML}{8B652B}
\begin{tikzpicture}[
  font=\fontsize{9}{11}\selectfont, text=SBink,
  box/.style={draw=SBline,line width=.6pt,rounded corners=2pt,
    align=center,inner xsep=6pt,inner ysep=6pt,text width=6.45cm},
  wide/.style={box,text width=14cm,fill=black!2},
  endpoint/.style={box,draw=SBblue,fill=SBbluefill},
  bridge/.style={box,draw=SBteal,fill=SBtealfill},
  extra/.style={box,draw=SBamber,fill=yellow!3,dashed},
  arrow/.style={-{Latex[length=2mm,width=1.4mm]},
    draw=SBline,line width=.7pt},
  row/.style={fit=#1,inner sep=0pt,outer sep=0pt}
]
\node[wide] (decomp) {
  \textbf{Schr\"odinger bridge: endpoint law + conditional path}
  \\[3pt]
  $\mathbb Q^*(x_{0:T})=\mathbb Q^*_{0,T}(a,b)\,
    \mathbb P_{a,b}(x_{0:T})$
  \quad Lemma~\ref{lem:dynamic-bridge-to-static-bridge}
};

\coordinate (fork) at ($(decomp.south)+(0,-.25)$);
\node[endpoint,anchor=north,minimum height=3.4cm] (static)
  at ($(fork)+(-3.78,-.25)$) {
  \textbf{Compute the endpoint coupling}\\[3pt]
  Quantum scaling: Theorem~\ref{th:quantum-static-bridge}
  \\[4pt]
  $A_{ij}=\mu_0(i)P^T(i,j)$
  \\[3pt]
  $\overline{\mathbb Q}_{0,T}(i,j)
    \propto A_{ij}e^{\bar x_i+\bar y_j}$
  \\[3pt]
  $d_{\mathrm{TV}}(\overline{\mathbb Q}_{0,T},
    \mathbb Q^*_{0,T})\le\varepsilon_s$
  \\[4pt]
  $\widetilde O(B^2N^{3/2}/\varepsilon_s^2)$
  \ quantum time
};

\node[bridge,anchor=north,minimum height=3.4cm] (conditional)
  at ($(fork)+(3.78,-.25)$) {
  \textbf{Construct the conditional sampler}\\[3pt]
  Local Gibbs updates: Theorem~\ref{cor:upper-bound-spectral-gap}
  \\[3pt]
  $\delta\ge\omega_-^4/(T^2N^3)$
  \\[4pt]
  Quantum walk: Theorem~\ref{th:complexity-gibbs-sampler}
  \\[3pt]
  $d_{\mathrm{TV}}(\rho_{a,b},\mathbb P_{a,b})\le\varepsilon_b$
  \\[4pt]
  {$O\!\left(\frac{TN^2}{\omega_-^{5/2}}
    \log\frac{N}{\omega_-\varepsilon_b}\right)$}
  \\[2pt] calls to $O_K$, $O_K^*$ and $S$
  \\[2pt] {plus $O(T\sqrt{N/\omega_-})$ calls to $U_P,U_P^*$}
};

\node[row={(static)(conditional)}] (methods) {};
\node[extra,anchor=north,minimum height=1.05cm] (endpoints)
  at ($(methods.south)+(-3.78,-.38)$) {
  \textbf{Draw endpoints $(a,b)$}\\[2pt]
  $O(N)$ extra work: Section~\ref{rem:static-output}
};
\node[bridge,anchor=north,minimum height=1.05cm] (sample)
  at ($(methods.south)+(3.78,-.38)$) {
  \textbf{Draw a path given $(a,b)$}\\[2pt]
  Herald success, then measure the path
};

\node[row={(endpoints)(sample)}] (sampling) {};
\node[wide,below=3.8mm of sampling] (output) {
  \textbf{Trajectory law:}\quad
  $\widetilde{\mathbb Q}(x_{0:T})
    =\widetilde{\mathbb Q}_{0,T}(a,b)\,\rho_{a,b}(x_{0:T})$
  \\[3pt]
  $d_{\mathrm{TV}}(\widetilde{\mathbb Q},\mathbb Q^*)
    \le 2\varepsilon_s+\varepsilon_b$
  on a successful scaling run
};

\draw[draw=SBline,line width=.7pt] (decomp.south) -- (fork);
\draw[arrow] (fork) -| (static.north);
\draw[arrow] (fork) -| (conditional.north);
\draw[arrow,dashed,draw=SBamber] (static.south) -- (endpoints.north);
\draw[arrow] (conditional.south) -- (sample.north);
\draw[arrow,dashed,draw=SBamber] (endpoints.east) -- (sample.west);
\draw[arrow] (sample.south) -- ++(0,-.19) -| (output.north);
\end{tikzpicture}
\endgroup
  \caption{Roadmap of the sampling procedure.
The dashed step samples $\widetilde{\mathbb Q}_{0,T}$
using the row procedure in Section~\ref{rem:static-output},
with additional work beyond computing the scaling vectors.
The path-error bound assumes a successful scaling run
and ideal endpoint sampling. The conditional sampler
requires coherent row preparation for $P$ and the bridge
Gibbs kernel. After failure, the entire conditional
procedure is repeated independently with the endpoints
fixed, requiring a constant expected number of attempts.}
  \label{fig:sb-algorithm-overview}
\end{figure}
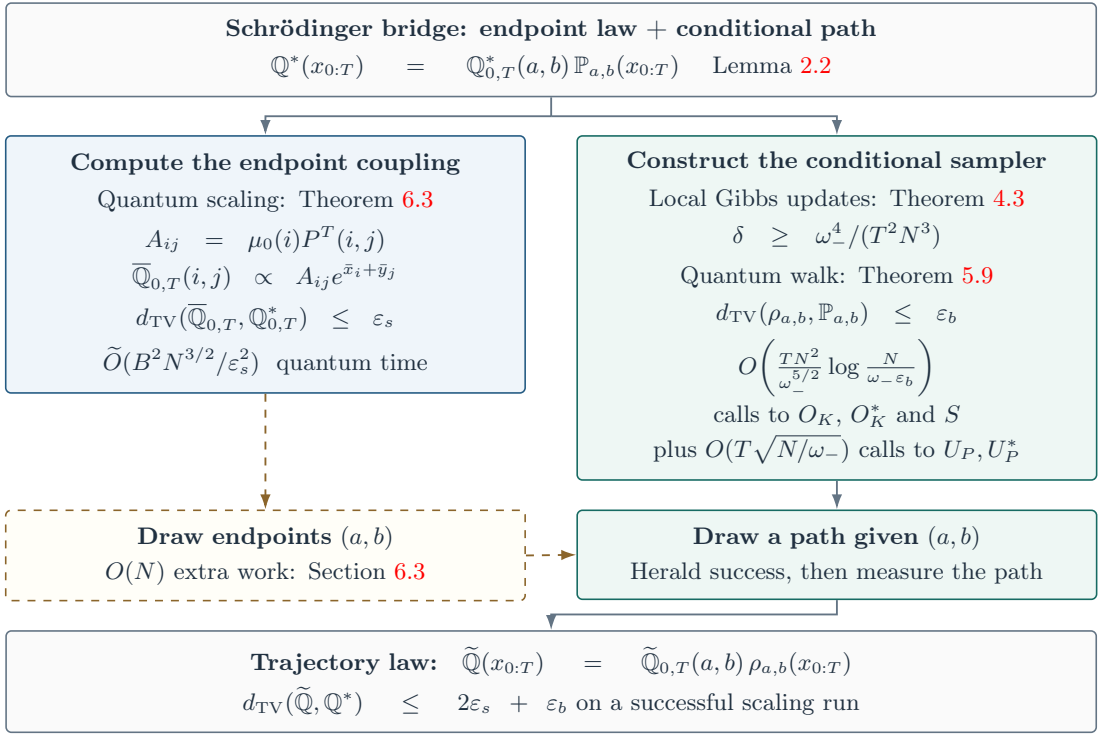

\paragraph{Bridges with additive running costs.}
In Section~\ref{sec:extension-with-additive-cost} we extend the analysis to trajectories that must meet the endpoint constraints while also favouring states with lower running costs. For a finite function $g:\chi\to[0,\infty)$, the objective adds the expected accumulated cost $\mathbb E_{\mathbb Q}[\sum_{t=0}^{T-1}g(X_t)]$ to the relative entropy. Proposition~\ref{prop-costs} rewrites this problem as a Schr\"odinger bridge with the normalised reference law
\[
    \widehat{\mathbb V}(x)
    =\frac{1}{Z_T}\mathbb P(x)
      \exp\!\left(-\sum_{t=0}^{T-1}g(x_t)\right),
\]
where $Z_T$ makes the total mass equal to one. This preserves the endpoint--conditional-path decomposition, but both components now involve the reweighted reference. The endpoint matrix of the reweighted reference law is obtained
from powers of $D_gP$, where
$D_g=\operatorname{diag}(e^{-g(i)})$. Applying the endpoint algorithm
requires checking that this new matrix satisfies the algorithm's
input-access and precision assumptions.

The main quantitative result is a gap estimate for local Gibbs updates under the reweighted conditional law $\widehat{\mathbb V}_{a,b}$. Let
$\Delta_g=\max_{i\in\chi}g(i)-\min_{i\in\chi}g(i)$.
Proposition~\ref{prop-rev-g} establishes reversibility and a nonnegative spectrum, and Theorem~\ref{cor:upper-bound-spectral-gap-g} proves
\[
    \delta_g\ge
    \frac{\omega_-^4}{T^2N^3}e^{-\Delta_g}.
\]
The weighted canonical-path argument controls the cost penalty through the range of a single-step cost, rather than the full accumulated range $T\Delta_g$. Consequently, for fixed $\Delta_g$, $N$, and $\omega_-$, the resulting relaxation-time bound remains quadratic in $T$. This gives quantitative control of local sampling even though the trajectory weights can differ exponentially with the horizon.

Theorem~\ref{th:complexity-gibbs-sampler-g} then gives a quantum
conditional-path sampler, provided the initial-state preparation $A_g$
can be implemented and reversed, the costed Gibbs updates are available
coherently, and the initial state has a known overlap of at least
$c_g>0$ with the target bridge state. The algorithm signals success
with probability bounded below by a universal constant. Conditional
on success, its output distribution has total variation error at most
$\varepsilon$. The resulting complexity bound is

\[
    O\!\left(
    \frac{TN^{3/2}e^{\Delta_g/2}}{c_g\omega_-^2}
    \log\frac{2}{c_g\varepsilon}
    \right)
\]
calls to the costed quantum-walk primitives, together with $O(1/c_g)$ calls to $A_g$ and its inverse. For fixed $N$, $\omega_-$, $\Delta_g$, and $c_g$, the displayed walk-query bound is linear in $T$. An overall running-time claim must also account for constructing and implementing the preparation and update routines.

A favourable spectral gap does not by itself guarantee an effective
quantum sampler: the initial state must also have sufficient overlap
with the target bridge state. Section~\ref{sec:extension-with-additive-cost} gives a two-state example
showing why initialisation matters. The usual initial state is prepared
without accounting for the running cost, whereas the target bridge
favours paths with lower accumulated cost. As the horizon grows, the
overlap between these two states decreases exponentially, even though
the range $\Delta_g$ of the single-step cost remains fixed. This small
overlap increases the cost of the quantum sampler. An initial state
that incorporates the running cost could avoid this difficulty, but
its preparation cost and overlap with the target would need to be
analysed separately.

Together with the classical bounds in
Remark~\ref{rem:cost-classical-comparison}, this example shows that
the running cost affects both the mixing of the Gibbs sampler
and the quality of the initial state. The quantum method retains
its square-root improvement in inverse spectral-gap dependence,
but this alone does not guarantee a speedup for arbitrary running costs.

\paragraph{Organisation of the paper.}
Section~\ref{sec:markov-chains-basics} reviews Markov-chain preliminaries.
Section~\ref{sec:setup-gibbs-sampler} introduces the classical Gibbs
sampler for Markov bridges and bounds its spectral gap.
Section~\ref{sec-q-gibbs} develops the quantum sampler and analyses its
complexity. Section~\ref{sec:quantum-sinkhorn-alg} addresses quantum
computation of the endpoint coupling, while
Section~\ref{sec:extension-with-additive-cost} treats the Schrödinger bridges with additive running costs. Sections~\ref{proof-d-s}--\ref{sec-pr-additive}
contain the proofs of the main results. Additional details on the numerical implementation of endpoint sampling are given in Appendix \ref{app:endpoint-precision}.

\section{Schrödinger Bridges} \label{sec:model}
\paragraph{Static Schrödinger Bridges.} 
Let $E$ be a finite state space equipped with its canonical $\sigma$-algebra.
For probability measures $\mathbb P$ and $\mathbb Q$ on $E$, the relative
entropy, or Kullback--Leibler (KL) divergence, is
\be \label{KL}
\operatorname{KL}(\mathbb Q\,\|\,\mathbb P)
=\begin{cases}
\displaystyle \mathbb E_{\mathbb Q}\!\left[
\log\frac{\mathrm d\mathbb Q}{\mathrm d\mathbb P}\right]
=\sum_{x\in E:\,\mathbb Q(x)>0}\mathbb Q(x)
\log\!\left(\frac{\mathbb Q(x)}{\mathbb P(x)}\right),
&\mathbb Q\ll\mathbb P,\\[4pt]
+\infty,&\mathbb Q\not\ll\mathbb P.
\end{cases}
\ee
Thus $\operatorname{KL}(\mathbb Q\,\|\,\mathbb P)$ measures the discrepancy
of $\mathbb Q$ from the reference law $\mathbb P$.

Let $\chi=\{1,\dots,N\}$.  Let $\mu$ and $\nu$ be probability measures on
$\chi$, let $p$ be a reference probability measure on $\chi^2$, and set
\[
\Pi(\mu,\nu)=\{q\in\mathcal P(\chi^2):q_0=\mu,\ q_1=\nu\}.
\]
We assume the finite-feasibility condition $\Pi(\mu,\nu)\cap\{q:q\ll p\}\ne\varnothing$.
The static Schrödinger bridge is then defined by 
\be \label{stat-br}
q^*=\underset{q\in\Pi(\mu,\nu)}{\arg\min}
\operatorname{KL}(q\,\|\,p).
\ee
The feasible set with finite KL divergence is compact, the KL divergence is
continuous on its supporting face, and it is strictly convex there; hence
$q^*$ exists and is unique.  Without the
finite-feasibility condition, all feasible couplings may have infinite KL
divergence, in which case the displayed optimisation does not specify a
unique bridge.

We write $p\sim\mu\otimes\nu$ for mutual absolute continuity.  The following
standard result identifies $q^*$ through Schrödinger potentials.

\begin{proposition}[cf.\ Proposition 1.8 in
\cite{tang2026foundationsschrodingerbridgesgenerative}]
\label{prop:static-potentials}
Let $\mu$ and $\nu$ be probability measures on the finite set
$\chi$, and let $p$ be a probability measure on $\chi\times\chi$
such that $p\sim\mu\otimes\nu$. Set
\[
    A=\operatorname{supp}(\mu),
    \qquad B=\operatorname{supp}(\nu).
\]
Then, there exist strictly positive functions
$\varphi^\circ:A\to(0,\infty)$ and
$\varphi^*:B\to(0,\infty)$ satisfying
\begin{equation}\label{eq:schrodinger-equation}
\begin{aligned}
\mu(x)
    &=\varphi^\circ(x)\sum_{z\in B}p(x,z)\varphi^*(z),
    && x\in A,\\
\nu(y)
    &=\varphi^*(y)\sum_{z\in A}p(z,y)\varphi^\circ(z),
    && y\in B.
\end{aligned}
\end{equation}
These functions are unique up to the reciprocal rescaling
\[
    (\varphi^\circ,\varphi^*)
    \mapsto(c\varphi^\circ,c^{-1}\varphi^*),
    \qquad c>0.
\]
Moreover, the unique minimiser of $\operatorname{KL}(\cdot\,\|\,p)$ in \eqref{stat-br} is given by 
\[
q^*(x,y)=
\begin{cases}
\varphi^\circ(x)p(x,y)\varphi^*(y),
    & (x,y)\in A\times B,\\
0,  & (x,y)\notin A\times B.
\end{cases}
\]
\end{proposition}

Consequently, computing the static bridge reduces to solving the
Schrödinger equations in \eqref{eq:schrodinger-equation}.

\paragraph{Dynamic Schrödinger Bridges.} 
Let $T\ge1$, let
$\Omega=\chi^{T+1}$, and write a path as
$x_{0:T}=(x_0,\dots,x_T)$.  Let $P$ be a Markov transition matrix on $\chi$
and let $\mu_0$ be the initial probability measure. We define $\{X_t\}_{t\geq0}$ as the Markov chain on $\chi$ with the transition matrix $P$. The reference path measure is defined as follows: 
\begin{equation}\label{eq:definition-probability-P}
\mathbb P(x_{0:T})=\mu_0(x_0)\prod_{t=0}^{T-1}P(x_t,x_{t+1}), \quad x_{0:T} \in \Omega. 
\end{equation}
Let $\mu_T$ be the prescribed final law and define
\be \label{lam-def}
\Lambda(\mu_0,\mu_T)
=\{\mathbb Q\in\mathcal P(\Omega):\mathbb Q_0=\mu_0,
\ \mathbb Q_T=\mu_T\}, 
\ee
where $\mathcal P(\Omega)$ is the set of probability measures on $\Omega$. 

We assume that there exists $\mathbb Q\in\Lambda(\mu_0,\mu_T)$ with
$\mathbb Q\ll\mathbb P$.  
Under this assumption the dynamic
Schrödinger bridge is the unique finite-KL minimiser
\be \label{dyn-br}
\mathbb Q^*=\underset{\mathbb Q\in\Lambda(\mu_0,\mu_T)}{\arg\min}
\operatorname{KL}(\mathbb Q\,\|\,\mathbb P).
\ee
For $a,b\in\chi$, define the path space, 
\be\label{eq:definition-Omega-a-b}
\Omega_{a,b}=\{x\in\Omega:x_0=a,\ x_T=b\}.
\ee
The endpoint law of the reference kernel is defined as follows:
\be \label{ref-stat}
\mathbb P_{0,T}(a,b)=\mathbb P(X_0=a,X_T=b)
=\mu_0(a)P^T(a,b).
\ee
In the following lemma, we establish the connection between the dynamic Schrödinger bridge, the Markov bridge, and the solution to the static Schrödinger bridge problem. For pairs with $\mathbb P_{0,T}(a,b)>0$, define the Markov bridge law, 
\be \label{m-bridge-def}
\mathbb P_{a,b}(\cdot)=\mathbb P(\,\cdot\mid X_0=a,X_T=b).  
\ee
\begin{lemma}[Factorisation of the Schrödinger bridge]\label{lem:dynamic-bridge-to-static-bridge}
Let $\mathbb P$ be any reference probability law on $\Omega$, and assume
that some $\mathbb Q\in\Lambda(\mu_0,\mu_T)$ satisfies
$\mathbb Q\ll\mathbb P$.  Let $\mathbb Q^*$ be the dynamic bridge in
\eqref{dyn-br}, and let $\mathbb Q^*_{0,T}$ be the static bridge with
marginals $(\mu_0,\mu_T)$ and reference law $\mathbb P_{0,T}$. Then, the following hold. 
\begin{itemize}
\item[\textbf{(i)}]
For every $(a,b)$ with $\mathbb Q^*_{0,T}(a,b)>0$,
\be \label{m-bridge}
\mathbb Q^*(\,\cdot\mid X_0=a,X_T=b)=\mathbb P_{a,b}(\cdot).
\ee
Consequently, for every $x\in\Omega$ whose endpoint pair is in the support
of $\mathbb P_{0,T}$,
\[
\mathbb Q^*(x)=\mathbb Q^*_{0,T}(x_0,x_T)\mathbb P_{x_0,x_T}(x).
\]
When $\mathbb P_{0,T}(x_0,x_T)=0$, the conditional law is not defined and the
displayed product is interpreted as zero; both sides are then zero.
\item[\textbf{(ii)}]
If, in addition, $\mathbb P_{0,T}\sim\mu_0\otimes\mu_T$, then
\[
\mathbb Q^*_{0,T}(a,b)=\varphi^\circ(a)\mathbb P_{0,T}(a,b)\varphi^*(b),
\qquad a,b\in\chi,
\]
where the potentials satisfy \eqref{eq:schrodinger-equation} on the
respective marginal supports.
Here each potential is extended by $1$ outside its
respective marginal support.
\item[\textbf{(iii)}]
Under the same assumptions in (ii),
\[
\mathbb Q^*(x)=\varphi^\circ(x_0)\mathbb P(x)\varphi^*(x_T),
\qquad x\in\Omega.
\]
\end{itemize}
\end{lemma}
The proof is given in Section~\ref{proof-d-s}.

\paragraph{The sampling scheme.}
Lemma~\ref{lem:dynamic-bridge-to-static-bridge} separates trajectory
sampling into an endpoint draw from $\mathbb Q^*_{0,T}$ and a conditional
draw from $\mathbb P_{a,b}$.  Section~\ref{sec-q-gibbs} develops a quantum
conditional-path sampler.  Section~\ref{sec:quantum-sinkhorn-alg} develops a
quantum procedure that computes an implicit approximation of the endpoint
coupling through its scaling potentials.  Producing an endpoint sample from
that implicit coupling requires an additional sampling routine and is not a
consequence of the endpoint-computation theorem alone. This routine is described in Section \ref{rem:static-output}. 
If the  Markov bridge quantum sampler reports failure, keep the sampled
endpoints fixed and repeat only the conditional Markov bridge sampler until success. Redrawing
endpoints after a failure could bias their distribution; the uniform
success bound gives constant expected repetition overhead.
  
\section{Essentials on Markov Chains} \label{sec:markov-chains-basics}
We recall some essential facts about finite-state Markov chains that will be used throughout this paper (see e.g. 
\cite{levin2017markov}).  Throughout this section, $P$ denotes a Markov
kernel on a finite state space $E$.

\begin{definition}{\textbf{(\cite{levin2017markov}, Section 1.3)}}{\label{def:irreducible-and-ergodic}}
The kernel $P$ is
\begin{enumerate}
\item[\textbf{(i)}] \emph{irreducible} if for every $x,y\in E$ there is an integer $t\ge1$
such that $P^t(x,y)>0$;
\item [\textbf{(ii)}] \emph{ergodic} if it is irreducible and $1$ is the only eigenvalue on
the unit circle, with algebraic multiplicity one.
\end{enumerate}
\end{definition}
For a finite irreducible chain, the second condition is equivalent to
aperiodicity.  An irreducible kernel has a unique stationary distribution
$\pi$, satisfying $\pi P=\pi$.

A Markov kernel $P$ is said to be reversible with respect to $\pi$ if it satisfies 
\be \label{dbc}
\pi(x)P(x,y)=\pi(y)P(y,x),\qquad \textrm{for all } x,y\in E.
\ee
Note that reversibility implies stationarity.

\begin{remark}{\label{rk:inner-product}}
One can define the following inner product for any functions $f,g:E\to\mathbb C$, 
\[
\langle f,g\rangle_\pi=\sum_{x\in E}\pi(x)\overline{f(x)}g(x).
\]
Then reversibility is equivalent to self-adjointness of $P$ with respect to this inner
product.
\end{remark}
We define the total variation distance for any two probability measures $\mu,\nu$ on $E$, 
\[
d_{\mathrm{TV}}(\mu,\nu)=\frac12\sum_{x\in E}|\mu(x)-\nu(x)|.
\]
If $P$ is ergodic then its iterates converge geometrically to its stationary measure $\pi$.
\begin{theorem}\label{th:convergence-stationnary-distrib}
There exist $C>0$ and $\alpha\in(0,1)$ such that, for every $t\ge1$,
\[
\max_{x\in E}d_{\mathrm{TV}}(P^t(x,\cdot),\pi)\le C\alpha^t.
\]
\end{theorem}
The $\varepsilon$-mixing time of $P$ quantifies the number of steps needed to get close to $\pi$.
\begin{definition}
For $\varepsilon\in(0,1)$, the $\varepsilon$-mixing time of $P$ is given by 
\[
\tau(\varepsilon)=\min\left\{t\in\{0,1,\dots\}:
\max_{x\in E}d_{\mathrm{TV}}(P^t(x,\cdot),\pi)\le\varepsilon\right\}.
\]
\end{definition}
The spectrum of $P$ often helps to derive bounds on the mixing time. 
In particular for a reversible ergodic kernel on $|E|\ge2$ states, order its real
eigenvalues as follows: 
\be \label{eigen-ord}
1=\lambda_1>\lambda_2\ge\cdots\ge\lambda_{|E|}>-1.
\ee

The spectral gap of $P$ is known as $1-\lambda_2$, but in this work we will use mostly the \emph{absolute spectral gap}, which is defined as, 
\[
\delta=1-\max\{|\lambda|:\lambda\in\sigma(P),\ \lambda\ne1\},
\]
and hence
\be \label{sp-rev}
\delta=1-(\lambda_2\vee|\lambda_{|E|}|).
\ee
The following theorem derives bounds on the mixing time in terms of the spectral gap. 
\begin{theorem}[\cite{levin2017markov}, Theorems~12.4 and~12.5] \label{thm-mix}
Let $P$ be reversible and ergodic, with stationary distribution $\pi$, and
let $\pi_{\min}=\min_{x\in E}\pi(x)$.  Then
\[
\left(\frac1\delta-1\right)\log\frac1{2\varepsilon}
\le\tau(\varepsilon)
\le\frac1\delta\log\frac1{\varepsilon\pi_{\min}}.
\]
\end{theorem}

The Metropolis--Hastings (MH) algorithm constructs a Markov kernel reversible with respect to a prescribed target distribution $\pi$ by combining a proposal kernel with an acceptance step. Let $\pi$ be strictly positive and let $K$ be a proposal kernel with symmetric support:
$K(x,y)>0$ if and only if $K(y,x)>0$. For $x\ne y$ with $K(x,y)>0$, set
\[
A(x,y)=\min\left\{1,
\frac{\pi(y)K(y,x)}{\pi(x)K(x,y)}\right\}.
\]
Set $A(x,y)=0$ when $x\ne y$ and $K(x,y)=0$.
The Metropolis--Hastings kernel is defined as follows: 
\[
Q(x,y)=
\begin{cases}
K(x,y)A(x,y),&x\ne y,\\[2pt]
\displaystyle 1-\sum_{z\ne x}K(x,z)A(x,z),&x=y.
\end{cases}
\]
This definition makes $Q$ stochastic and gives
$\pi(x)Q(x,y)=\pi(y)Q(y,x)$, so $Q$ is reversible with stationary distribution $\pi$.

A single-coordinate Gibbs update is a special case of MH: it proposes a new value for one coordinate from its full conditional distribution under $\pi$, keeping the other coordinates fixed. For this proposal, the MH ratio equals one, so every proposal is accepted. Choosing the coordinate at random with fixed probabilities gives a reversible Gibbs kernel. 
\section{Gibbs sampler for Markov bridges} \label{sec:setup-gibbs-sampler}
Fix $N\ge2$, $T\ge2$, and endpoints $a,b\in\chi$. Let $P$ be an ergodic Markov kernel on $\chi$ with initial distribution $\mu_0$. 
\begin{assumption} \label{assump-spec} 
We assume throughout this section that $a\in\operatorname{supp}(\mu_0)$ and that, for some
$\omega_->0$,
\be \label{om}
\min_{i,j\in\chi}P(i,j)\ge\frac{\omega_-}{N},
\qquad
\min_{x\in\operatorname{supp}(\mu_0)}\mu_0(x)
\ge\frac{\omega_-}{N}.
\ee
\end{assumption} 
The first condition is a uniform positivity, or Doeblin, condition.  It
implies $0<\omega_-\le1$ and ensures
$\mathbb P(X_0=a,X_T=b)>0$.  The lower bound on $\mu_0$ is recorded for
later complexity estimates; it is not needed for the conditional Gibbs-gap
bound.  No reversibility assumption on the underlying matrix $P$ is needed
for the results of this section.

The Markov bridge law can be written as follows: 
\be \label{m-br}
\mathbb P_{a,b}(x)=\mathbb P(x\mid X_0=a,X_T=b)
=\frac{\mathbb P(x)}{\mathbb P(X_0=a,X_T=b)},
\qquad x\in\Omega_{a,b}.
\ee
For $1\le t\le T-1$, define 
\be \label{x-minus}
x_{-t}=(x_0,\dots,x_{t-1},x_{t+1},\dots,x_T).
\ee
We construct a Gibbs sampler on $\Omega_{a,b}$. Using the Markov property we conclude that the full conditional law of the $t$-th coordinate under the bridge depends only
on its two neighbours and is
\begin{equation}\label{eq:conditionned-interior-proba-P}
c_t(s\mid x_{-t})
:=\mathbb P_{a,b}(X_t=s\mid X_{-t}=x_{-t})
=\frac{P(x_{t-1},s)P(s,x_{t+1})}
{\sum_{z\in\chi}P(x_{t-1},z)P(z,x_{t+1})},
\qquad s\in\chi.
\end{equation}
The kernel of the Gibbs sampler on $\Omega_{a,b}$ is defined as follows: 
\begin{equation}\label{eq:definition-Gibbs-sampler}
K(x,y)=\frac1{T-1}\sum_{t=1}^{T-1}
\mathds 1_{\{x_{-t}=y_{-t}\}}c_t(y_t\mid x_{-t}).
\end{equation}
Thus one step chooses an interior coordinate uniformly and resamples it from
its conditional probability.  Distinct paths are connected in one step precisely
when they differ at one interior coordinate. Paths differing at several
coordinates are joined by a sequence of such updates, as described in the following proposition. We will need some further definitions before we state this result. 

Let
\[
    \Omega_{a,b}^{+}
    =\{\mathbf{x}\in\Omega_{a,b}:
        \mathbb P_{a,b}(\mathbf{x})>0\}.
\]
We define $L^2(\mathbb P_{a,b})$ as the space of functions
$f:\Omega_{a,b}^{+}\to\mathbb C$, equipped with the inner product
\[
    \langle f,g\rangle_{L^2(\mathbb P_{a,b})}
    =\sum_{\mathbf{x}\in\Omega_{a,b}^{+}}
      \mathbb P_{a,b}(\mathbf{x})\,
      \overline{f(\mathbf{x})}\,g(\mathbf{x}),
    \qquad f,g\in L^2(\mathbb P_{a,b}).
\]
The associated squared norm is
\[
    \|f\|_{L^2(\mathbb P_{a,b})}^{2}
    =\sum_{\mathbf{x}\in\Omega_{a,b}^{+}}
      \mathbb P_{a,b}(\mathbf{x})\,|f(\mathbf{x})|^2,
\]
which is finite since $\Omega_{a,b}^{+}$ is finite.

\begin{proposition}\label{prop-rev}
Let $P$ be an ergodic Markov kernel on $\chi$ with initial distribution $\mu_0$ satisfying $P(i,j)>0$ for all $i,j\in\chi$ and
$\mu_0(a)>0$. Then the kernel $K$ has the following properties:
\begin{itemize}
\item[\textbf{(i)}] It is irreducible and aperiodic on $\Omega_{a,b}$.
Moreover, $K(x,x)>0$ for every $x$.
\item[\textbf{(ii)}] It is reversible with respect to $\mathbb P_{a,b}$.
\item[\textbf{(iii)}] As an operator on
$L^2(\mathbb P_{a,b})$, it is positive semidefinite.
\end{itemize}
\end{proposition}
The proof is given in Section~\ref{sec-pf-gibbs}.

We use the canonical-paths method (see Section~13.4 of
\cite{levin2017markov}) to control the  spectral gap
$1-\lambda_2$.  Proposition~\ref{prop-rev}(iii) shows that $K$ has no
negative eigenvalues, so its spectral gap and absolute spectral gap \eqref{sp-rev} coincide. Our main result in this direction is stated below and proved in Section~\ref{sec-pf-gibbs}.

\begin{theorem}\label{cor:upper-bound-spectral-gap}
Let $P$ be an ergodic Markov kernel on $\chi$ with initial distribution $\mu_0$ satisfying Assumption \ref{assump-spec}. 
Then, the absolute spectral gap $\delta$ of the Gibbs kernel in \eqref{eq:definition-Gibbs-sampler} satisfies
\[
\delta\ge\frac{\omega_-^4}{T^2N^3}.
\]
\end{theorem}
 
\section{Quantum Gibbs sampler}\label{sec-q-gibbs}
In this section we construct a quantum sampler for the conditional bridge $\pi:=\mathbb P_{a,b}$ using the Gibbs kernel $K$ from \eqref{eq:definition-Gibbs-sampler}. We continue to assume that $N\ge2$ and $T\ge2$, and that $P$ is an ergodic Markov kernel on $\chi$ with initial distribution $\mu_0$. 
The construction uses the qubitised-walk framework of
\cite{claudon2026quantumcircuitsmetropolishastingsalgorithm,Szegedy}. We provide some essential definitions and results of this framework in the following. 

\begin{definition}[Projected unitary encodings]\label{def:PUE-and-SPUE}
Let $U$ be a unitary on a finite-dimensional Hilbert space
$\mathcal K$, and let $\square_L,\square_R:\mathcal H\to\mathcal K$
be isometries. The triple $(U,\square_L,\square_R)$ is a
\emph{projected unitary encoding} (PUE) of
$A=\square_L^*U\square_R$. If $\square_L=\square_R=\square$
and $U=U^*$, the pair $(U,\square)$ is a
\emph{symmetric projected unitary encoding} (SPUE).
\end{definition}
Definition~\ref{def:PUE-and-SPUE} uses isometries. The more general
definition of a PUE using partial isometries requires restriction
to their initial support spaces.

For a PUE of $A$, define its Hermitianisation by
\begin{equation}\label{op-u-sq}
\overline U=
\begin{pmatrix}0&U\\U^*&0\end{pmatrix},
\qquad
\overline\square=
\begin{pmatrix}\square_L&0\\0&\square_R\end{pmatrix}.
\end{equation}
\begin{proposition}[{\cite[Proposition~1]{claudon2026quantumcircuitsmetropolishastingsalgorithm}}]
The pair $(\overline U,\overline\square)$ is an SPUE of
$\overline A=\begin{pmatrix}0&A\\A^*&0\end{pmatrix}$.
The eigenvalues of $\overline A$ are the positive and negative
singular values of $A$, with multiplicities.
The unitary $\overline U$ can be implemented using controlled-$U$,
controlled-$U^*$, and a Pauli $X$ gate.
\end{proposition}

For an SPUE $(U,\square)$, the associated quantum walk is
\begin{equation}\label{q-walk}
    \mathcal W=(2\square\square^*-I)U.
\end{equation}
The following theorem describes how the eigenvalues of $A$
determine those of $\mathcal W$.

\begin{theorem}[{\cite[Theorem~2]{claudon2026quantumcircuitsmetropolishastingsalgorithm}}]
\label{thm:SPUE-spectral}
Let $(U,\square)$ be an SPUE of $A$, with $\square$ an isometry.
Suppose that $A|v\rangle=\lambda|v\rangle$ and $\|v\|=1$.
If $-1<\lambda<1$, set $\vartheta=\arccos\lambda$.
Then $\mathcal W$ has eigenvalues $e^{\pm i\vartheta}$,
with corresponding normalised eigenvectors
\be \label{eig-ph} 
    |\mu(\lambda)^\pm\rangle
    =\frac{(e^{\pm i\vartheta}I-U)\square|v\rangle}
           {\sqrt{2}\sin\vartheta}.
\ee
If $\lambda=\pm1$, then $\square|v\rangle$ is an eigenvector
of $\mathcal W$ with eigenvalue $\lambda$.
\end{theorem}

Theorem \ref{thm:SPUE-spectral} explains how a small eigenvalue gap becomes a
larger separation between eigenphases. If the largest
eigenvalue of $A$ is $1$ and the second-largest is $1-\delta$,
then the corresponding eigenvalues of $\mathcal W$ are $1$
and $e^{\pm i\vartheta}$, where
\[
    \vartheta=\arccos(1-\delta)
    \sim\sqrt{2\delta}
    \qquad\text{as }\delta\downarrow0.
\]
Thus, an eigenvalue gap of size $\delta$ for $A$ becomes
a phase separation of order $\sqrt{\delta}$ for $\mathcal W$.
This relation is the spectral basis for the improvement
in gap dependence established below.

In the following $\operatorname{ran}$ denotes the range of an operator.
The eigenvectors in the theorem describe $\mathcal W$ on
the subspace
\[
    \mathcal L
    =\operatorname{span}\bigl(
        \operatorname{ran}\square,\,
        U\operatorname{ran}\square
      \bigr),
\]
consisting of linear combinations of vectors
$\square|v\rangle$ and $U\square|v\rangle$.
This subspace is invariant under $\mathcal W$: applying
$\mathcal W$ to a vector in $\mathcal L$ produces another
vector in $\mathcal L$. As $|v\rangle$ ranges over an
orthonormal eigenbasis of $A$, the eigenvectors constructed
in the theorem form an orthonormal basis of $\mathcal L$.
They therefore give the complete spectral decomposition
of $\mathcal W|_{\mathcal L}$.

To describe the action on the orthogonal complement,
let $|w\rangle\in\mathcal L^\perp$.
Since $U^2=I$, the operator $U$ preserves $\mathcal L$.
Its unitarity then implies that it also preserves
$\mathcal L^\perp$. Consequently, $U|w\rangle$ is
orthogonal to the range of $\square$, and hence
$\square\square^*U|w\rangle=0$. It follows that
\[
    \mathcal W|w\rangle
    =(2\square\square^*-I)U|w\rangle
    =-U|w\rangle.
\]
Thus, $\mathcal W=-U$ on $\mathcal L^\perp$.
Since $U^2=I$, the restriction of $\mathcal W$ to this
subspace can only have eigenvalues $1$ or $-1$.

\paragraph{The bridge walk and its accessible subspace.}
Let $\mathcal H=\mathbb C^{\Omega_{a,b}}$, of dimension $N^{T-1}$,
and let $S$ be the swap operator: for any $x,y\in\Omega_{a,b}$, 
\begin{equation}\label{u-sw}
S|x,y\rangle=|y,x\rangle,\qquad S=S^*=S^{-1}.
\end{equation}
For $K$ the Markov kernel in \eqref{eq:definition-Gibbs-sampler} define the isometry $\square:\mathcal H\to\mathcal H\otimes\mathcal H$ by
\begin{equation}\label{step}
\square|x\rangle
=\sum_{y\in\Omega_{a,b}}\sqrt{K(x,y)}\,|x,y\rangle.
\end{equation}

A direct calculation gives
\[
D:=\square^*S\square,\qquad \textrm{for } \ 
D(x,y)=\sqrt{K(x,y)K(y,x)}.
\]
By reversibility (Proposition~\ref{prop-rev}),
$\pi(x)K(x,y)=\pi(y)K(y,x)$. Since $\pi$ is strictly
positive, we obtain
\[
    D(x,y)
    =\sqrt{K(x,y)K(y,x)}
    =\frac{\sqrt{\pi(x)}}{\sqrt{\pi(y)}}K(x,y).
\]
Consequently,
\[
    D=\operatorname{diag}(\sqrt{\pi})\,
      K\,\operatorname{diag}(\sqrt{\pi})^{-1}.
\]
Thus $D$ is Hermitian and has the same eigenvalues as $K$.
We use its qubitised walk
\begin{equation}\label{W-Q}
\mathcal W=(2\square\square^*-I)S
\end{equation}
on the invariant subspace
\[
\mathcal L=
\operatorname{span}(\operatorname{ran}\square,
S\operatorname{ran}\square).
\]
Define:
\begin{equation}\label{pab}
|\pi\rangle=|\mathbb P_{a,b}\rangle
=\sum_{x\in\Omega_{a,b}}\sqrt{\pi(x)}\,|x\rangle.
\end{equation}

\begin{proposition}\label{cor:SPUE-spectral}
Under the assumptions of Proposition~\ref{prop-rev},
let $\pi=\mathbb P_{a,b}$ and let $\delta$ denote the
absolute spectral gap of $K$. Then, the following hold:
\begin{enumerate}
\item[\textbf{(i)}]
The vector $\square|\pi\rangle$ is the unique normalised
$1$-eigenvector of $\mathcal W$ in $\mathcal L$,
up to multiplication by a scalar of modulus one.

\item[\textbf{(ii)}]
The smallest absolute value of a nonzero eigenphase of
$\mathcal W|_{\mathcal L}$, with phases chosen in
$[-\pi,\pi]$, is
\[
    \Delta
    =\arccos\!\left(\max(\sigma(K)\setminus\{1\})\right)
    =\arccos(1-\delta)
    \ge\sqrt{2\delta}.
\]
\item[\textbf{(iii)}] Let $\Pi$ be the orthogonal projector onto
$\ker(\mathcal W-I)$. For every $|\psi\rangle\in\mathcal L$,
\be \label{proj}
    \Pi|\psi\rangle
    =\square|\pi\rangle
      \langle\pi|\square^*|\psi\rangle.
\ee
Thus, on $\mathcal L$, $\Pi$ acts as the rank-one
projector onto $\square|\pi\rangle$.
\end{enumerate}
\end{proposition}
 
Items (i) and (ii) in Proposition \ref{cor:SPUE-spectral} follow immediately from reversibility, ergodicity of $K$ derived in Proposition~\ref{prop-rev}, together with 
Theorem~\ref{thm:SPUE-spectral}. To understand why (iii) holds note that by Proposition \ref{cor:SPUE-spectral}  the $1$-eigenspace of $\mathcal W|_{\mathcal L}$ is $\operatorname{span}\{\square|\pi\rangle\}$ so \eqref{proj} describes the projection on this subspace. 

\paragraph{Coherent access to the Gibbs update.}
We assume access to a unitary $O_K$ satisfying
\begin{equation}\label{o-k}
O_K(|0\rangle\otimes|x\rangle)
=\sum_{y\in\Omega_{a,b}}\sqrt{K(x,y)}\,|x,y\rangle,
\qquad x\in\Omega_{a,b}.
\end{equation}
We also assume access to its adjoint $O_K^*$.
Here the first input register is workspace and the second encodes
a bridge path; fixed endpoint registers may be omitted.
Equivalently, 
\begin{equation}\label{sq-o-k}
\square=O_K(|0\rangle\otimes I).
\end{equation}
The isometry property ensures that such a unitary extension exists,
but does not by itself bound its implementation cost. In this oracle
model $\mathcal W$ in \eqref{W-Q} can be written as 
\begin{equation}\label{w-o-op}
\mathcal W=
O_K\bigl((2|0\rangle\langle0|-I)\otimes I\bigr)O_K^*S.
\end{equation}
A controlled application of $\mathcal W$ uses one call to
$O_K$ and one to $O_K^*$, together with a controlled swap
and a controlled reflection about the zero subspace.

\subsection{The sampling algorithm and complexity analysis }
We present the quantum Gibbs sampler for Markov bridges and derive its complexity.   
We first prepare a state supported on the prescribed bridge endpoints.

\paragraph{A fixed-endpoint initial state.}
In addition to $O_K$, we assume coherent access to the original
transition kernel $P$ through a unitary $U_P$ satisfying 
\begin{equation}\label{eq:coherent-reference-transition}
U_P(|i\rangle|0\rangle)
=|i\rangle\sum_{j\in\chi}\sqrt{P(i,j)}\,|j\rangle, \quad i \in \chi. 
\end{equation}

We also assume access to the adjoint $U_P^*$.
Let $A_{a,b}$ denote the unitary circuit that prepares
the initial path state. Starting from zero-initialised
registers, it writes $a$ into the first register, applies
$U_P$ successively to generate the $T-1$ intermediate
states, each retained in a separate register, and writes
$b$ into the final register. Thus, $A_{a,b}$ simulates
the reference chain up to time $T-1$ and fixes the final
endpoint without simulating the last transition to $b$.
The resulting state is
\begin{equation}\label{sigma}
A_{a,b}|0\rangle=|\Sigma_{a,b}\rangle
=\sum_{x\in\Omega_{a,b}}\sqrt{q_{a,b}(x)}\,|x\rangle,
\qquad
q_{a,b}(x)=\prod_{t=0}^{T-2}P(x_t,x_{t+1}),
\end{equation}
where $|0\rangle$ denotes the joint zero state of all
registers. Measuring $|\Sigma_{a,b}\rangle$ in the path
basis yields a trajectory with distribution $q_{a,b}$.

The circuit $A_{a,b}$ uses $T-1$ calls to $U_P$, and
its inverse uses $T-1$ calls to $U_P^*$. Since the
endpoints are known constants, this preparation requires
neither the endpoint probability $P^T(a,b)$ nor the
preparation of a state encoding the initial marginal.

We illustrate the action of $A_{a,b}$ for $T=3$. 
In this case the preparation proceeds as follows: 
\[
\begin{aligned}
|0,0,0,0\rangle
&\longmapsto |a,0,0,0\rangle\\
&\longmapsto
\sum_{x_1\in\chi}\sqrt{P(a,x_1)}
\,|a,x_1,0,0\rangle\\
&\longmapsto
\sum_{x_1,x_2\in\chi}
\sqrt{P(a,x_1)P(x_1,x_2)}
\,|a,x_1,x_2,0\rangle\\
&\longmapsto
\sum_{x_1,x_2\in\chi}
\sqrt{P(a,x_1)P(x_1,x_2)}
\,|a,x_1,x_2,b\rangle.
\end{aligned}
\]
The two middle steps apply $U_P$ to successive pairs
of registers. The first and last steps write the known
endpoints into zero-initialised registers using Pauli
$X$ gates. In particular, fixing the final register to
$b$ does not simulate a transition to $b$, so the
amplitudes contain no factor $\sqrt{P(x_2,b)}$.

Define the overlap
\begin{equation}\label{over}
\theta_{a,b}
=\langle\square\pi|\square\Sigma_{a,b}\rangle
=\langle\pi|\Sigma_{a,b}\rangle.
\end{equation}

\begin{lemma}\label{lem-ol}
Let $P$ be a Markov kernel on the finite state space $\chi$, and fix $T\ge2$ and $a,b\in\chi$.
Suppose that, for some $\eta>0$,
\[
   P(i,b)\ge\eta,
\qquad
\text{for every }i\in\chi\text{  such that  }P^{T-1}(a,i)>0.
\]
Let $\pi=\mathbb P_{a,b}$ denote the bridge law for the
reference chain started at $a$. Then the distribution
$q_{a,b}$ in \eqref{sigma} satisfies
\[
    q_{a,b}(x)\le\eta^{-1}\pi(x),
    \qquad x\in\Omega_{a,b},
\]
and the overlap in \eqref{over} satisfies
\[
    \theta_{a,b}\ge\sqrt{\eta}.
\]
In particular, if the transition bound in \eqref{om} holds, then
$\theta_{a,b}\ge\sqrt{\omega_-/N}$.
\end{lemma}

The proof is given in Section~\ref{sec-pf-qsamp}. Note that this explicit lower
bound is independent of the bridge length.

\paragraph{Amplitude amplification with a known lower bound.}
The initial state $\square|\Sigma_{a,b}\rangle$ has a
nonzero overlap with the target state $\square|\pi\rangle$.
We use amplitude amplification
(see, e.g., \cite{brassardamplitudeamplification})
to increase the probability of successfully preparing
an approximation to this target state, using a known
lower bound on the initial overlap, given in Lemma \ref{lem-ol}. 

We first describe amplitude amplification for a general
quantum preparation procedure. Let $\mathcal A$ be a unitary
circuit that prepares a normalised state
$|\Psi\rangle=\mathcal A|0\rangle$, where $|0\rangle$ denotes
the joint zero state of all registers. Here, $\mathcal A$
denotes the preparation circuit to which amplification is
applied; it need not be the initial-path circuit $A_{a,b}$
alone.

Let $\Pi_*$ be the orthogonal projector onto the subspace
corresponding to successful outcomes, called the
\emph{marked subspace}. Decompose the prepared state as
\begin{equation}\label{psi-dec}
|\Psi\rangle=|\Psi_*\rangle+|\Psi_\perp\rangle,
\qquad
|\Psi_*\rangle=\Pi_*|\Psi\rangle,
\qquad
|\Psi_\perp\rangle=(I-\Pi_*)|\Psi\rangle.
\end{equation}
These two components are orthogonal and need not be
normalised. The initial success probability is
$\|\Psi_*\|^2$. Amplitude amplification increases this
probability by applying a suitable number of iterations
of the operator
\begin{equation}\label{r0}
Q=-\mathcal A S_0\mathcal A^*S_*,
\end{equation}
called the Grover iterate, where
\begin{equation}\label{r1}
S_0=I-2|0\rangle\langle0|,
\qquad
S_*=I-2\Pi_*.
\end{equation}
The operator $S_*$ changes the sign of the successful
component and leaves its orthogonal complement unchanged.
Similarly, $S_0$ changes the sign of the all-zero state.
Since $\mathcal A|0\rangle=|\Psi\rangle$ we get 
\begin{equation}\label{r2}
\mathcal A S_0\mathcal A^*
=I-2|\Psi\rangle\langle\Psi|.
\end{equation}
Thus, the Grover iterate combines a sign change on the
marked subspace with a reflection about the prepared
state $|\Psi\rangle$.

The following proposition specifies the amplitude-amplification
guarantee used in our sampling algorithm. Given only a known
lower bound $c$ on the amplitude of the successful component,
it provides a procedure with success probability bounded below
by a universal constant and cost $O(1/c)$. In particular,
the exact initial overlap need not be known.

\begin{proposition}\label{prop:amplitude-amplification}
Suppose that a lower bound $c\in(0,1]$ satisfying
\begin{equation}\label{l-psi}
    \|\Psi_*\|\ge c
\end{equation}
is known, and that $\mathcal A$, its inverse
$\mathcal A^*$, and controlled implementations of the
reflection $S_*$ are available.
Then there exists an algorithm that uses one preparation of $|\Psi\rangle$
and $O(1/c)$ controlled applications of $Q$ and returns a
nonnegative integer $n=O(1/c)$ such that, with probability
at least $8/\pi^2$,
\[
    \|\Pi_*Q^n|\Psi\rangle\|\ge\frac25.
\]
\end{proposition}
In Section~\ref{sec:amplitude-amplification} we prove Proposition \ref{prop:amplitude-amplification} and in particular give the phase-estimation procedure and derive the stated constants.

\begin{remark} \label{rem-grov} 
From Proposition \ref{prop:amplitude-amplification} it follows that preparing a fresh copy of $|\Psi\rangle$, applying $Q^n$,
and measuring whether the resulting state lies in the
marked subspace therefore succeeds with probability at
least $32/(25\pi^2)$, where the probability includes
both the estimation and final measurement outcomes.
Across both stages, the total number of calls to
$\mathcal A$, $\mathcal A^*$, and the marking reflection
$S_*$ (including its controlled version) is $O(1/c)$.
\end{remark}

\paragraph{Approximating the stationary-state projector.}
We now construct an approximation to the orthogonal
projector $\Pi$ in \eqref{proj}. This projector extracts the stationary component of
the initial state. Since
$\square|\Sigma_{a,b}\rangle\in\mathcal L$, we have from \eqref{over}, 
\[
    \Pi\square|\Sigma_{a,b}\rangle
    =\theta_{a,b}\square|\pi\rangle.
\]
Thus, applying $\Pi$ and normalising gives the target
state $\square|\pi\rangle$.

Since $\Pi$ is generally not unitary, we cannot implement
it directly as a unitary circuit. Instead, we approximate
it by a polynomial $\Gamma(\mathcal W)$ and implement
this polynomial as part of a unitary acting on one
additional qubit. Measuring that qubit and obtaining
$0$ applies $\Gamma(\mathcal W)$ to the input state,
up to normalisation. This outcome signals successful
application of the approximation. The following proposition bounds the approximation error
and the number of controlled applications of $\mathcal W$
required for this construction. We first introduce some auxiliary definitions and notation.

Let $\mathcal W$ be unitary, and let $\Pi$ denote the
orthogonal projector onto $\ker(\mathcal W-I)$.
We call $\Delta_0\in(0,\pi]$ a lower bound on the
\emph{phase gap} if every eigenvalue $e^{i\theta}\ne1$
of $\mathcal W$, with $\theta\in[-\pi,\pi]$, satisfies
$|\theta|\ge\Delta_0$.

A polynomial $\Phi$ is \emph{complementary} to a polynomial $\Gamma$ if
\[
    |\Gamma(z)|^2+|\Phi(z)|^2=1,
    \qquad\text{for all }|z|=1.
\]
This identity allows the two polynomials to describe
the two branches of a unitary circuit with an additional
qubit, as follows.

\begin{proposition}\label{prop-poly}
Suppose that a lower bound $\Delta_0$ on the phase gap
of $\mathcal W$ is known. For every $0<\eta<1$, there
exists a polynomial $\Gamma$ of degree
$O(\Delta_0^{-1}\log(2/\eta))$ such that
\begin{equation}\label{gm-w}
    \|\Gamma(\mathcal W)-\Pi\|\le\eta,
    \qquad
    \sup_{|z|=1}|\Gamma(z)|\le1,
\end{equation}
where $\|\cdot\|$ denotes the operator norm.
Moreover, there exist a complementary polynomial $\Phi$
and a unitary $V$, using one additional qubit such that
\be \label{v-op} 
    V(|0\rangle\otimes|\psi\rangle)
    =|0\rangle\otimes\Gamma(\mathcal W)|\psi\rangle
     +|1\rangle\otimes\Phi(\mathcal W)|\psi\rangle,
\ee
for every input state $|\psi\rangle$.
The circuit uses $O(\Delta_0^{-1}\log(2/\eta))$
controlled applications of $\mathcal W$.
\end{proposition}

For a normalised input $|\psi\rangle$, measuring the
additional qubit yields outcome $0$ with probability
$\|\Gamma(\mathcal W)|\psi\rangle\|^2$.
Conditional on obtaining this outcome, the remaining
registers are in the state
\[
    \frac{\Gamma(\mathcal W)|\psi\rangle}
         {\|\Gamma(\mathcal W)|\psi\rangle\|}.
\]
Thus, outcome $0$ signals successful application of
the approximate projector $\Gamma(\mathcal W)$.

 The construction in Proposition~\ref{prop-poly} follows
\cite{baptistereflectioneigenspace}; its proof is given
in Section~\ref{sec:reflection-eigenspace}.
By Proposition~\ref{cor:SPUE-spectral}, we may take
$\Delta_0=\sqrt{2\delta_0}$ for any known positive lower
bound $\delta_0$ on the absolute spectral gap of $K$.
Combining the initial-state preparation and overlap bound
(Lemma~\ref{lem-ol}), amplitude amplification
(Proposition~\ref{prop:amplitude-amplification}), and
the approximation of the stationary-state projector
(Proposition~\ref{prop-poly}) yields the following
complexity bound for the quantum Gibbs sampler. 

\begin{theorem}[Complexity of the quantum Gibbs sampler]
\label{th:complexity-gibbs-sampler}
Let $P$ be a Markov kernel on $\chi$ with initial
distribution $\mu_0$ satisfying Assumption~\ref{assump-spec}.
Assume that $N,T\ge2$, that the positivity parameter
$\omega_->0$ in \eqref{om} is known, and that access to
the primitives \eqref{o-k} and
\eqref{eq:coherent-reference-transition}, together with
their adjoints, is available.

For every $0<\varepsilon<1$, there exists a quantum
algorithm that heralds success with probability bounded
below by a universal constant. Conditional on success,
measuring its path register yields a distribution $\rho$
on $\Omega_{a,b}$ satisfying
\[
    d_{\mathrm{TV}}(\rho,\mathbb P_{a,b})\le\varepsilon.
\]
The algorithm uses
\begin{equation}\label{comp-gibbs}
    O\!\left(
        \frac{TN^2}{\omega_-^{5/2}}
        \log\frac{N}{\omega_-\varepsilon}
    \right)
\end{equation}
calls to the walk primitives $O_K,O_K^*,S$, and
\[
    O\!\left(T\sqrt{\frac{N}{\omega_-}}\right)
\]
calls to $U_P,U_P^*$.
\end{theorem}
The proof in Section~\ref{pf-gibbs} describes the full
algorithm, including the initial-state preparation circuit
and its inverse used in each amplification step.

\begin{remark}[Classical comparison]
\label{rem:classical-gibbs-comparison}
For a classical Gibbs sampler initialised from $q_{a,b}$,
Lemma~\ref{lem-ol} and
Theorem~\ref{cor:upper-bound-spectral-gap} imply that
\[
    O\!\left(
        \frac{T^2N^3}{\omega_-^4}
        \log\frac{N}{\omega_-\varepsilon}
    \right)
\]
Gibbs updates suffice to achieve a total variation error
of at most $\varepsilon$. Initialisation requires an
additional $T-1$ transitions of the reference chain.
The classical bound is quadratic in $T$, whereas the
bound on quantum walk calls is linear. This comparison
concerns upper bounds for the two Gibbs procedures under
their respective access models; it does not establish
an advantage over all classical bridge samplers.
\end{remark}

\section{Quantum computation of the static Schr\"odinger bridge} \label{sec:quantum-sinkhorn-alg}
In this section we provide a quantum algorithm for computing the Schr\"odinger potentials $\varphi^\circ,\varphi^*$ that scale the reference endpoint distribution $\P_{0,T}$ in Lemma~\ref{lem:dynamic-bridge-to-static-bridge}(ii).

From \eqref{ref-stat} we have
\be \label{eq:ref-endpoint-matrix}
\P_{0,T}(i,j)=\mu_0(i)P^T(i,j),
\qquad i,j\in\chi.
\ee
In particular, $\P_{0,T}$ is a nonnegative $N\times N$
matrix satisfying
\[
    \sum_{i,j\in\chi}\P_{0,T}(i,j)=1.
\]
We assume coherent exact query access to the entries of
$P^T$, $\mu_0$, and $\mu_T$. Under
Assumption~\ref{assump-box}(iv),
\eqref{eq:ref-endpoint-matrix} allows each entry of
$\P_{0,T}$ to be computed using a constant number of
input-oracle calls and polylogarithmic arithmetic overhead.
Throughout this section, $P^T$ is supplied as part of
the input, consistent with the exact-entry oracle model
of \citet{Improved-quantum-lower-and-upper-bounds-for-matrix-scaling}.
The arithmetic cost of constructing $P^T$ is discussed
at the end of the section.

\subsection{Quantum matrix scaling} \label{sec-box}

Following \cite{Improved-quantum-lower-and-upper-bounds-for-matrix-scaling}, we introduce logarithmic scaling variables $x,y\in\mathbb R^N$ and define the nonnegative matrix
\be \label{eq:scaled-endpoint-matrix}
M(x,y)_{ij}:=\P_{0,T}(i,j)e^{x_i+y_j},
\qquad i,j\in\chi,
\ee
and the convex matrix-scaling potential
\be \label{eq:static-convex-potential}
f(x,y):=
\sum_{i,j\in\chi}\P_{0,T}(i,j)e^{x_i+y_j}
-\langle\mu_0,x\rangle-\langle\mu_T,y\rangle.
\ee
In the following we assume  that $\P_{0,T}\sim\mu_0\otimes\mu_T$.  If
$(x^*,y^*)$ minimises $f$, then the first-order conditions give
\be \label{lg1} 
\sum_{j\in\chi}M(x^*,y^*)_{ij}=\mu_0(i),
\qquad
\sum_{i\in\chi}M(x^*,y^*)_{ij}=\mu_T(j).
\ee
Consequently, Proposition \ref{prop:static-potentials} implies that
\be \label{eq:static-bridge-matrix-scaling}
\Q_{0,T}^*(i,j)
=M(x^*,y^*)_{ij}
=\P_{0,T}(i,j)e^{x_i^*+y_j^*},
\ee
and the Schr\"odinger potentials are
$\varphi^\circ=e^{x^*}$ and $\varphi^*=e^{y^*}$.

The method introduced in \cite[Algorithm~1 and Theorem~3.13]{Improved-quantum-lower-and-upper-bounds-for-matrix-scaling} approximately minimises the function $f$ in \eqref{eq:static-convex-potential}, by a regularised, box-constrained Newton method. At each iteration, quantum amplitude estimation approximates the
gradient, while quantum graph sparsification constructs a sparse
approximation of the Hessian, exploiting its Laplacian structure
after a sign change. A classical subroutine approximately minimises
the resulting quadratic model over an $\ell_\infty$-ball, yielding
a damped update of the logarithmic potentials. An additional
rescaling step controls the total mass of the scaled matrix.

More precisely, let $N=|\chi|$ and let $\mathbb{P}_{0,T}$ be the normalised
endpoint reference matrix. Define 
\[
m
:=
\#\bigl\{(i,j)\in\chi^2:
          \mathbb{P}_{0,T}(i,j)>0\bigr\}.
\]
For target probability vectors $\mu_0,\mu_T$, define $M(x,y)$ and $f$ as in \eqref{eq:scaled-endpoint-matrix} and \eqref{eq:static-convex-potential}, respectively, and let
 \[
f^\star=\inf_{x,y\in\mathbb{R}^{N}}f(x,y).
\]
To apply
\cite[Algorithm~1 and Theorem~3.13]
{Improved-quantum-lower-and-upper-bounds-for-matrix-scaling},
we impose the following assumption.

\begin{assumption}  \label{assump-box}
\begin{enumerate}
\item[\textbf{(i)}] \emph{Positive marginals and quantitative lower bounds.}
The probability vectors $\mu_0$ and $\mu_T$ are strictly
positive. There exist constants $a,c>0$ not depending on $N$ and $T$ such that
\[
\min_{\substack{i,j\in\chi\\
                 \mathbb{P}_{0,T}(i,j)>0}}
\mathbb{P}_{0,T}(i,j)\ge aN^{-c},
\qquad
\min_{i\in\chi}\mu_0(i)\ge aN^{-c},
\qquad
\min_{j\in\chi}\mu_T(j)\ge aN^{-c}.
\]
\item[\textbf{(ii)}]  \emph{Asymptotic scalability.}
For every $\delta>0$, there exist $x,y\in\mathbb{R}^{N}$ such that
\[
\|M(x,y)\mathbf{1}-\mu_0\|_1\le\delta,
\qquad
\|M(x,y)^\top\mathbf{1}-\mu_T\|_1\le\delta.
\]

\item[\textbf{(iii)}] \emph{Bounded approximate minimiser.}
For the prescribed optimisation accuracy $\eta\in(0,1)$,
a bound $B\ge1$ is available such that there exist
$x_\eta,y_\eta\in\mathbb{R}^{N}$ satisfying
\[
\|(x_\eta,y_\eta)\|_\infty\le B,
\qquad
f(x_\eta,y_\eta)-f^\star\le\eta^2.
\]
Here
$\|(x,y)\|_\infty
=\max\{\|x\|_\infty,\|y\|_\infty\}$.
In particular, a bound on the norm of an exact minimiser
suffices.
\item[\textbf{(iv)}]
\emph{Input access and computational model.}
We consider entrywise-positive endpoint matrices
$\mathbb P_{0,T}$, whose support is the known set
$\chi\times\chi$. We assume coherent entry-query access
to exact rational representations of
$\mathbb P_{0,T}$, $\mu_0$, and $\mu_T$, with input
encoding lengths bounded by $\operatorname{polylog}(N)$.

We use the quantum-access and QCRAM model of the cited
quantum box-constrained Newton algorithm, allowing
classical updates and coherent queries to the scaling
vectors. Oracle calls, individual memory operations,
and arithmetic at the required working precision incur
at most polylogarithmic overhead in $N$, $B$, and $1/\eta$.
The costs of constructing the endpoint kernel and
preparing its oracle representation are excluded.
The encoding-length bound is an additional restriction
on the input family; it does not follow from the
assumptions on the one-step transition matrix.
\end{enumerate}

\end{assumption}

Under Assumption \ref{assump-box}, \cite[Theorem~3.13]
{Improved-quantum-lower-and-upper-bounds-for-matrix-scaling}
guarantees that Algorithm~1 therein returns classical
vectors $x,y\in\mathbb{R}^{N}$ satisfying
\[
f(x,y)-f^\star\le6\eta^2
\]
with probability at least $2/3$, in quantum time
$
\widetilde O\!\left(
\eta^{-2}B^2 N^{3/2}
\right), 
$
where for an entrywise-positive reference matrix we have $m=N^2$. The output specifies $M(x,y)$ implicitly through its
logarithmic potentials. These bounds are conditional on the stated encoding,
access, and memory assumptions. The costs of constructing
$\mathbb{P}_{0,T}$ and preparing its oracle representation
are additional.  

\subsection{Quantum algorithm for Schr\"odinger potentials}
To verify Assumption~\ref{assump-box}(iii), we first establish an $\ell_\infty$ bound for a minimiser of the convex potential in \eqref{eq:static-convex-potential}.

\begin{lemma}\label{lem:static-scaling-bound}
Assume that $\P_{0,T}$, $\mu_0$, and $\mu_T$ are entrywise positive, and set
\[
p_{\min}:=\min_{i,j\in\chi}\P_{0,T}(i,j),\qquad
\mu_{0,\min}:=\min_{i\in\chi}\mu_0(i),\qquad
\mu_{T,\min}:=\min_{j\in\chi}\mu_T(j).
\]
There exists a minimiser $(x^*,y^*)$ of $f$ such that
\be \label{eq:static-scaling-bound}
\lVert(x^*,y^*)\rVert_\infty
\le B_0
:=\log\left(
\frac{N}{p_{\min}\mu_{0,\min}\mu_{T,\min}}
\right).
\ee
\end{lemma}
The proof of Lemma \ref{lem:static-scaling-bound} is given in Section \ref{sec-pf-stat}.

\begin{algorithm}[t]
\caption{Quantum computation of the Schr\"odinger potentials}
\label{alg:quantum-static-bridge}
\begin{algorithmic}[1]
\Require Exact coherent entry access to $P^T$, $\mu_0$, and $\mu_T$;
an error $\eps\in(0,1)$; a bound $B\ge\max\{1,B_0\}$.
\State Implement the exact entry oracle for $\P_{0,T}$ using
\eqref{eq:ref-endpoint-matrix}.
\State Set $\eta=\eps/\sqrt{3}$.
\State Apply the quantum box-constrained Newton algorithm of
\cite[Algorithm 1 and Theorem 3.13]
{Improved-quantum-lower-and-upper-bounds-for-matrix-scaling}
to the matrix $\P_{0,T}$, with target row sums $\mu_0$, target column sums
$\mu_T$, accuracy parameter $\eta$, and diameter bound $B$.
\State Let $(\bar x,\bar y)$ denote the returned scaling vectors and set
$\bar\varphi^\circ=e^{\bar x}$ and $\bar\varphi^*=e^{\bar y}$.
\State \Return $(\bar x,\bar y)$.
\end{algorithmic}
\end{algorithm}

The next theorem bounds the computational cost and the error in the endpoint coupling specified by the computed potentials.
 
\begin{theorem}[Quantum computation of the static Schr\"odinger bridge]
\label{th:quantum-static-bridge}
Suppose that Assumptions \ref{assump-box}(i) and (iv) hold and that
$\P_{0,T}$, $\mu_0$, and $\mu_T$ are entrywise positive.
Let $\eps\in(0,1)$ and choose
$B\ge\max\{1,B_0\}$.
Then, with probability at
least $2/3$, Algorithm \ref{alg:quantum-static-bridge} returns scaling
vectors $(\bar x,\bar y)$ for which the normalised matrix defined by
\be \label{eq:approx-static-bridge}
\overline M_{ij}:=\P_{0,T}(i,j)e^{\bar x_i+\bar y_j},\qquad
L:=\sum_{i,j\in\chi}\overline M_{ij},\qquad
\overline{\Q}_{0,T}(i,j):=\frac{\overline M_{ij}}{L}
\ee
satisfies
\be \label{eq:static-bridge-TV-bound}
d_{TV}(\overline{\Q}_{0,T},\Q_{0,T}^*)\le\eps.
\ee
The quantum running time is
\be \label{eq:static-bridge-quantum-complexity-general}
\widetilde O\left(\frac{B^2 N^{3/2}}{\eps^2}\right).
\ee
\end{theorem}
The proof of Theorem \ref{th:quantum-static-bridge} is given in Section \ref{sec-pf-stat}. 
\begin{remark}\label{rem:static-uniform-positivity}

For the simplified complexity bound below, assume in addition the following lower bounds on both prescribed marginals:
\be \label{eq:static-uniform-positivity}
P(i,j)\ge\frac{\omega_-}{N},\qquad
\mu_0(i)\ge\frac{\omega_-}{N},\qquad
\mu_T(j)\ge\frac{\omega_-}{N}, \quad \textrm{ for all } i,j \in \chi.
\ee
This implies 
\[
P^T(i,j)
=\sum_{k\in\chi}P^{T-1}(i,k)P(k,j)
\ge\frac{\omega_-}{N},  \quad \textrm{ for all } i,j \in \chi, \ T \geq 1. 
\]
It follows from \eqref{eq:ref-endpoint-matrix} that
$p_{\min}\ge\omega_-^2/N^2$, and therefore recalling \eqref{eq:static-scaling-bound}, 
$B_0=O(\log(N/\omega_-))$.  
For fixed $\omega_->0$, Theorem~\ref{th:quantum-static-bridge} therefore gives the bound $\widetilde O(\eps^{-2}N^{3/2})$. Here and in the matrix-scaling bounds, $\widetilde O$ suppresses polylogarithmic factors in the dimension and the inverse accuracy. If $\omega_-$ varies, the dependence through $B=O(\log(N/\omega_-))$ must also be retained.
\end{remark} 

\begin{remark}[Interpretation of the speedup]\label{rem:static-speedup}
For a dense endpoint matrix, the classical box-constrained
Newton method has running time $\widetilde O(BN^2)$,
whereas Algorithm \ref{alg:quantum-static-bridge} has
running time $\widetilde O(B^2N^{3/2}/\eps^2)$.
Thus, at fixed accuracy and when
$B=\operatorname{polylog}(N)$, the dependence on the
state-space size improves from $N^2$ to $N^{3/2}$,
up to logarithmic factors.
More generally, comparing these bounds and suppressing
logarithmic factors gives an asymptotic improvement
when $\eps^2\gg B/\sqrt N$.
This comparison is conditional on the stated input-access
and computational assumptions and excludes the costs of
constructing $P^T$ and preparing the endpoint oracle.
\end{remark}

\subsection{From scaling vectors to endpoint samples}
\label{rem:static-output}

We now explain how to use the computed scaling vectors
to sample an endpoint pair and bound its error relative
to the optimal coupling $\Q_{0,T}^*$.
The algorithm returns two vectors, each with $N$ entries,
which specify $\overline{\Q}_{0,T}$ through
\eqref{eq:approx-static-bridge}. Constructing the full
matrix would require $\Omega(N^2)$ operations.
The procedure below evaluates only one row.

Fix the returned vectors $(\bar x,\bar y)$.
First, sample $X_0$ from the prescribed marginal $\mu_0$.
Given $X_0=i$, evaluate the $N$ weights
$\P_{0,T}(i,j)e^{\bar y_j}$, normalise them, and sample
$X_T$ from
\be \label{row} 
    q_i(j)
    :=\frac{\P_{0,T}(i,j)e^{\bar y_j}}
            {\sum_{k\in\chi}\P_{0,T}(i,k)e^{\bar y_k}},
    \qquad j\in\chi.
\ee
The denominator is positive under the hypotheses of
Theorem~\ref{th:quantum-static-bridge}.
This is the conditional distribution of the second
endpoint given the first under $\overline{\Q}_{0,T}$:
the row factor $e^{\bar x_i}$ and the global normalisation
constant cancel. Thus, no other row needs to be evaluated.

The resulting joint distribution is
\begin{equation}\label{eq:sampled-static-bridge}
    \widetilde{\Q}_{0,T}(i,j)=\mu_0(i)q_i(j).
\end{equation}
Its first marginal is exactly $\mu_0$, while its second
marginal may differ from $\mu_T$. In particular,
$\widetilde{\Q}_{0,T}$ need not equal
$\overline{\Q}_{0,T}$, because we sample the first
endpoint from $\mu_0$ rather than from the first marginal
of the computed coupling.

To bound the resulting error, write
$\bar\mu_0(i)=\sum_j\overline{\Q}_{0,T}(i,j)$.
Since $\overline{\Q}_{0,T}(i,j)=\bar\mu_0(i)q_i(j)$,
the two distributions have the same conditional laws,
and hence
\[
    d_{\mathrm{TV}}(
        \widetilde{\Q}_{0,T},\overline{\Q}_{0,T})
    =d_{\mathrm{TV}}(\mu_0,\bar\mu_0)
    \le d_{\mathrm{TV}}(
        \Q_{0,T}^*,\overline{\Q}_{0,T}).
\]
The inequality follows because taking a marginal cannot
increase total variation distance. The triangle
inequality therefore gives
\begin{equation}\label{eq:sampled-static-bridge-TV}
    d_{\mathrm{TV}}(\widetilde{\Q}_{0,T},\Q_{0,T}^*)
    \le
    2d_{\mathrm{TV}}(\overline{\Q}_{0,T},\Q_{0,T}^*).
\end{equation}

Consequently, running the endpoint algorithm with
tolerance $\varepsilon/2$ produces, with probability
at least $2/3$, scaling vectors whose associated sampler
has total variation error at most $\varepsilon$.
This guarantee applies to the sampling distribution
for each successful scaling run. Since success of the
scaling computation is not assumed to be detectable,
the same error bound does not automatically hold when
averaging over all scaling runs.

\paragraph{A sharper error bound.}
The sampling procedure also preserves the relative-entropy
guarantee obtained in the proof of
Theorem~\ref{th:quantum-static-bridge}.
Indeed, replacing the approximate first marginal $\bar\mu_0$
by the prescribed marginal $\mu_0$ can only reduce the
relative entropy from the target coupling. The chain rule for relative entropy gives
\begin{equation}\label{eq:sampled-static-bridge-KL}
\begin{aligned}
\operatorname{KL}(\Q_{0,T}^*\|\widetilde{\Q}_{0,T})
&=\operatorname{KL}(\Q_{0,T}^*\|\overline{\Q}_{0,T})
  -\operatorname{KL}(\mu_0\|\bar\mu_0)\\
&\le\operatorname{KL}(\Q_{0,T}^*\|\overline{\Q}_{0,T}).
\end{aligned}
\end{equation}
Let $\varepsilon_s\in(0,1)$ denote the accuracy parameter
used when running the endpoint algorithm in
Theorem~\ref{th:quantum-static-bridge}.
Its proof shows that, with probability at least $2/3$,
\[
    \operatorname{KL}(\Q_{0,T}^*\|\overline{\Q}_{0,T})
    \le2\varepsilon_s^2;
\]
see \eqref{eq:static-objective-error}--\eqref{eq:static-objective-KL-identity}.
Combining this bound with
\eqref{eq:sampled-static-bridge-KL} and Pinsker's
inequality gives
\[
    d_{\mathrm{TV}}(\widetilde{\Q}_{0,T},\Q_{0,T}^*)
    \le
    \sqrt{\frac12
      \operatorname{KL}(\Q_{0,T}^*\|\widetilde{\Q}_{0,T})}
    \le\varepsilon_s.
\]
Thus, the relative-entropy estimate in the proof yields
endpoint error at most $\varepsilon_s$, improving on
the bound $2\varepsilon_s$ obtained from
\eqref{eq:sampled-static-bridge-TV}.

\paragraph{Sampling cost.}
Sampling from $\mu_0$ and then evaluating and sampling
the selected row in \eqref{row} requires $O(N)$ entry
evaluations and arithmetic operations. Working-precision
costs are accounted for under
Assumption~\ref{assump-box}(iv); the costs of constructing
$P^T$ and its oracle representation are excluded.

Let $\widehat\mu_0$ and $\widehat q_i$ denote the
distributions actually sampled after numerical
approximation. The resulting joint endpoint law is
\[
    \widehat{\Q}_{0,T}(i,j)
    :=\widehat\mu_0(i)\widehat q_i(j).
\]
If
\[
    d_{\mathrm{TV}}(\widehat\mu_0,\mu_0)\le\xi_0,
    \qquad
    \sup_{i\in\chi}d_{\mathrm{TV}}(\widehat q_i,q_i)\le\xi_1,
\]
then, whenever
$d_{\mathrm{TV}}(\overline{\Q}_{0,T},\Q_{0,T}^*)
\le\varepsilon_s$,
\begin{equation}\label{eq:sampled-static-bridge-precision}
    d_{\mathrm{TV}}(\widehat{\Q}_{0,T},\Q_{0,T}^*)
    \le2\varepsilon_s+\xi_0+\xi_1.
\end{equation}
Under the relative-entropy guarantee above,
$2\varepsilon_s$ can be replaced by $\varepsilon_s$.
The numerical implementation and its precision
requirements are detailed in
Appendix~\ref{app:endpoint-precision}.

\paragraph{From endpoint pairs to complete trajectories.}
To sample a complete trajectory, first draw an endpoint
pair $(i,j)$ from $\widehat{\Q}_{0,T}$, then apply the
quantum conditional-path sampler of
Theorem~\ref{th:complexity-gibbs-sampler} with these
endpoints. Let $\rho_{i,j}$ denote its output law
conditional on success, with accuracy chosen so that
\[
    \sup_{i,j}
    d_{\mathrm{TV}}(\rho_{i,j},\P_{i,j})\le\zeta.
\]

If the sampler reports failure, repeat the entire
conditional-path procedure with fresh randomness,
keeping the endpoints fixed. Since each attempt
succeeds with probability bounded below by a universal
constant, this procedure terminates almost surely
after a constant expected number of attempts.
Its output law for the fixed endpoints remains
$\rho_{i,j}$. Resampling the endpoints after failure
could bias their distribution, since the success
probability may depend on $(i,j)$.

Let $\widehat{\Q}$ denote the resulting distribution
on complete trajectories. It satisfies
\[
    \widehat{\Q}(x_{0:T})
    =\widehat{\Q}_{0,T}(x_0,x_T)
      \rho_{x_0,x_T}(x_{0:T}).
\]
By Lemma~\ref{lem:dynamic-bridge-to-static-bridge},
the target distribution $\Q^*$ has endpoint law
$\Q_{0,T}^*$ and conditional path laws $\P_{i,j}$.
The total error is therefore bounded by the sum
of the endpoint error and the conditional-path error:
\begin{equation}\label{eq:static-to-dynamic-sampling-error}
\begin{aligned}
    d_{\mathrm{TV}}(\widehat{\Q},\Q^*)
    &\le d_{\mathrm{TV}}(
        \widehat{\Q}_{0,T},\Q_{0,T}^*)+\zeta\\
    &\le2\varepsilon_s+\xi_0+\xi_1+\zeta,
\end{aligned}
\end{equation}
where the second inequality uses
\eqref{eq:sampled-static-bridge-precision}.

To obtain total variation error at most $\varepsilon$,
choose the accuracy parameters so that
$2\varepsilon_s+\xi_0+\xi_1+\zeta\le\varepsilon$.
When endpoint sampling introduces no numerical error
($\xi_0=\xi_1=0$), it suffices to take
$\varepsilon_s=\varepsilon/4$ and $\zeta=\varepsilon/2$.

This guarantee applies to each scaling output satisfying
the error bound in
Theorem~\ref{th:quantum-static-bridge}, an event of
probability at least $2/3$. Unlike the success of the
conditional-path sampler, success of the scaling
computation is not assumed to be detectable.
The total cost comprises the scaling computation,
endpoint sampling, and all attempts of the
conditional-path sampler.

\subsection{Classical computation of the endpoint kernel}

When only the one-step transition matrix $P$ is provided,
the endpoint algorithm requires an additional computation
of $P^T$. We describe two classical methods and their
arithmetic costs.

\paragraph{Repeated squaring.}
The simplest method computes the powers
$P,P^2,P^4,\ldots$ by successively squaring the preceding
matrix. Writing $T$ in binary,
\[
    T=\sum_{\ell=0}^{L}b_\ell2^\ell,
    \qquad b_\ell\in\{0,1\},
    \qquad L=\lfloor\log_2T\rfloor,
\]
we then multiply the powers corresponding to the
nonzero binary digits:
\[
    P^T=\prod_{\ell:\,b_\ell=1}P^{2^\ell}.
\]
This requires $O(\log(T+1))$ matrix multiplications.
If $\mathsf{MM}(N)$ denotes the arithmetic cost of
multiplying two $N\times N$ matrices, the total cost is
\[
    O\!\left(\mathsf{MM}(N)\log(T+1)\right).
\]
With standard matrix multiplication, this becomes
$O(N^3\log(T+1))$. Only the current power and the
accumulated product need to be stored, together with
the multiplication workspace, giving an $O(N^2)$
storage bound in scalar entries.

\paragraph{Computation via Frobenius normal form.}
An alternative uses the Frobenius normal form, which
exists for every square matrix and does
not require diagonalizability.
The algorithm of \citet{storjohannVillard2000}
computes this form and the associated change of basis
in $O(N^3)$ field operations.

In this basis, computing powers reduces to computing
polynomial remainders $z^T\bmod f_\ell(z)$, where
the polynomials $f_\ell$ are determined by $P$ and
their degrees sum to $N$. Polynomial repeated squaring
and reconstruction of the blocks require
$O(N^2\log(T+1))$ operations. Including the changes
of basis, the total cost is therefore
\[
    O\!\left(N^3+N^2\log(T+1)\right)
\]
field operations.

As with repeated squaring, this is an exact-arithmetic
bound. For rational inputs, bit complexity also depends
on the encoding lengths of intermediate and output
entries.

\section{Dynamic Schrödinger bridge with additive cost}\label{sec:extension-with-additive-cost}
We now extend the dynamic Schr\"odinger bridge problem
to include a running cost, allowing us to penalize
undesirable states along the trajectory while preserving
the prescribed endpoint marginals. The reference path law $\mathbb P$ and the marginal constraint set $\Lambda(\mu_0,\mu_T)$ remain as in \eqref{eq:definition-probability-P} and \eqref{lam-def}.

\begin{definition}[DSB with costs]
Let $g:\chi\to\mathbb R_+$ be a finite cost function and let $\mu_0,\mu_T$ be probability distributions on $\chi$. The dynamic Schr\"odinger bridge with additive cost is
\be\label{cost-add}
\mathbb Q^g
=\underset{\mathbb Q\in\Lambda(\mu_0,\mu_T)}{\arg\min}
\left\{\mathbb E_{\mathbb Q}\left[\sum_{t=0}^{T-1}g(X_t)\right]
+\operatorname{KL}(\mathbb Q\Vert\mathbb P)\right\}.
\ee
\end{definition}

Exponential reweighting absorbs the running cost into a new reference measure. Its normalisation must be included when using relative entropy between probability distributions. This is described in the following proposition. 
\begin{proposition}\label{prop-costs}
Let $P$ be a Markov kernel on the finite state space
$\chi$, let $\mu_0,\mu_T$ be probability distributions
on $\chi$, and let $g:\chi\to\mathbb R$.
Assume that $\Lambda(\mu_0,\mu_T)$ contains at least
one law $\mathbb Q_0$ with $\mathbb Q_0\ll\mathbb P$.
Define the finite measure
\be\label{eq:definition-of-R}
\mathbb V(x)
=\mu_0(x_0)\prod_{t=0}^{T-1}
P(x_t,x_{t+1})e^{-g(x_t)},\qquad x\in\Omega,
\ee
and let
\be \label{v-hat} 
Z_T=\sum_{x\in\Omega}\mathbb V(x)>0,
\qquad \widehat{\mathbb V}(x)=\mathbb V(x)/Z_T.
\ee
For every $\mathbb Q\in\Lambda(\mu_0,\mu_T)$,
\be\label{eq:cost-entropy-identity}
\mathbb E_{\mathbb Q}\left[\sum_{t=0}^{T-1}g(X_t)\right]
+\operatorname{KL}(\mathbb Q\Vert\mathbb P)
=\operatorname{KL}(\mathbb Q\Vert\widehat{\mathbb V})-\log Z_T.
\ee
Consequently, $\mathbb Q^g$ is the dynamic Schr\"odinger bridge with reference probability distribution $\widehat{\mathbb V}$.
\end{proposition}
The proof is given in Section~\ref{sec-pr-additive}.

\paragraph{The weighted reference and its endpoint matrix.}
We absorb the running cost into the reference path law
by multiplying each path probability by the exponential
of its negative accumulated cost and then normalising.
This reduces the problem to the same endpoint-coupling
and conditional-path decomposition as in Lemma \ref{lem:dynamic-bridge-to-static-bridge}, with
a modified reference distribution.

Set
\be \label{W-def} 
D_g=\operatorname{diag}(e^{-g(i)})_{i\in\chi},
\qquad W=D_gP.
\ee
The matrix $W$ is a transfer matrix with row sums $e^{-g(i)}$, rather than a stochastic transition matrix. Its powers give the endpoint reference distribution:
\be\label{eq:cost-endpoint-reference}
\widehat{\mathbb V}_{0,T}(a,b)
=\frac{\mu_0(a)(W^T)(a,b)}{Z_T},
\qquad Z_T=\mu_0^\top W^T\mathbf1.
\ee
For endpoints with positive reference probability, write
$\widehat{\mathbb V}_{a,b}(\cdot)=\widehat{\mathbb V}(\cdot\mid X_0=a,X_T=b)$.

Recall that $\mathbb Q^g$ is the minimiser of problem \eqref{cost-add}. 
Its joint endpoint distribution is denoted by
$\mathbb Q^g_{0,T}$ and is defined by
\[
    \mathbb Q^g_{0,T}(a,b)
    :=\mathbb Q^g(X_0=a,X_T=b)
    =\sum_{\substack{x\in\Omega\\x_0=a,\;x_T=b}}
      \mathbb Q^g(x).
\]
By Lemma~\ref{lem:dynamic-bridge-to-static-bridge},
$\mathbb Q^g_{0,T}$ solves the static bridge problem
with reference $\widehat{\mathbb V}_{0,T}$.
Moreover, for every endpoint pair $(a,b)$ with
$\mathbb Q^g_{0,T}(a,b)>0$,
\[
    \mathbb Q^g(\,\cdot\mid X_0=a,X_T=b)
    =\widehat{\mathbb V}_{a,b}(\cdot).
\]
In other words, once the endpoints are fixed, the
optimal path law agrees with the conditional law
of the weighted reference.

To apply Algorithm~\ref{alg:quantum-static-bridge},
the weighted endpoint matrix $\widehat{\mathbb V}_{0,T}$
and the prescribed marginals must satisfy the assumptions
of Theorem~\ref{th:quantum-static-bridge}.
In particular, we require positivity, the quantitative
lower bounds in Assumption~\ref{assump-box}(i), and
the exact rational representation and coherent access
specified in Assumption~\ref{assump-box}(iv).

These requirements are not automatic for general
running costs. Even when $g$ is rational-valued,
the weights $e^{-g(i)}$ may be irrational.
We must therefore restrict the inputs to those satisfying
the theorem's assumptions, or separately analyse an
implementation using approximate entries.

The cost of constructing $\widehat{\mathbb V}_{0,T}$
and its oracle representation must also be accounted for.
Once the scaling vectors have been computed, an endpoint
pair can be sampled by the procedure in
Section~\ref{rem:static-output}, with the corresponding
additional cost.

The normalised path law $\widehat{\mathbb V}$ in
\eqref{v-hat} can be generated by a Markov chain,
but its transition probabilities generally are time dependent. Recall $W$ from \eqref{W-def}, and define
\[
    h_T(i)=1,
    \qquad
    h_t(i)=\sum_{j\in\chi}W(i,j)h_{t+1}(j),
    \qquad 0\le t<T.
\]
Here, $h_t(i)$ is the total weight of all continuations
from state $i$ at time $t$ to the final time $T$.
Equivalently,
\[
    h_t(i)
    =\mathbb E_P\!\left[
        \exp\!\left(-\sum_{s=t}^{T-1}g(X_s)\right)
        \,\middle|\,X_t=i
      \right],
\]
where the expectation refers to the reference chain $P$
started at $i$ at time $t$.

The corresponding initial distribution and transition
probabilities of a Markov chain with law $\widehat{\mathbb V}$  are given by 
\begin{equation}\label{eq:cost-reference-markov-representation}
    \nu_0(i)=\frac{\mu_0(i)h_0(i)}{Z_T},
    \qquad
    P_t^g(i,j)=\frac{W(i,j)h_{t+1}(j)}{h_t(i)}.
\end{equation}
For finite-valued $g$ the recursion ensures that
$\sum_jP_t^g(i,j)=1$.
The factors involving $h_t$ cancel along each path,
giving
\[
    \nu_0(x_0)\prod_{t=0}^{T-1}P_t^g(x_t,x_{t+1})
    =\frac{\mu_0(x_0)}{Z_T}
      \prod_{t=0}^{T-1}W(x_t,x_{t+1})
    =\widehat{\mathbb V}(x).
\]

Thus, incorporating the running cost changes both
the initial distribution and the transition probabilities:
in general, $\nu_0\ne\mu_0$.
Simply multiplying row $i$ of $P$ by $e^{-g(i)}$
does not produce a stochastic matrix.  

For dense input, computing all $h_t$ by the backward
recursion requires $O(TN^2)$ arithmetic operations.
A quantum state-preparation method based on
\eqref{eq:cost-reference-markov-representation}
must account for this preprocessing cost, as well
as the costs of preparing $\nu_0$ and coherently
implementing the transitions $P_t^g$.

\paragraph{The Gibbs sampler with fixed endpoints.}
We now turn to sampling the weighted reference law
conditional on fixed endpoints. Its single-coordinate
conditional distributions depend only on $P$, the running
cost, and the two neighbouring states, allowing us to
define local Gibbs updates without computing the weights $h_t$.

For the remaining sampling results, assume $N\ge2$, $T\ge2$, $P(i,j)\ge\omega_-/N>0$, for all $i,j\in \chi$ and $a\in\operatorname{supp}(\mu_0)$. Similarly to \eqref{eq:definition-Gibbs-sampler}
 we define the Gibbs kernel, 
\be\label{eq:cost-gibbs-kernel}
K^g(x,y)
=\frac1{T-1}\sum_{t=1}^{T-1}
\mathds1_{\{x_{-t}=y_{-t}\}}
q^g(y_t,x_{t-1},x_{t+1}),\qquad x,y\in\Omega_{a,b},
\ee
where
\be\label{eq:cost-gibbs-conditional}
q^g(z,\ell,r)
=\frac{P(\ell,z)P(z,r)e^{-g(z)}}
{\sum_{v\in\chi}P(\ell,v)P(v,r)e^{-g(v)}}.
\ee
Only the cost at the updated coordinate appears in this conditional probability; all other cost factors cancel.
The following Proposition is the analog of Proposition \ref{prop-rev} for DSB with costs. 

\begin{proposition}\label{prop-rev-g}
The kernel $K^g$ is irreducible and aperiodic on $\Omega_{a,b}$ and is reversible with respect to $\widehat{\mathbb V}_{a,b}$. Each update has a positive probability of leaving the
path unchanged and a positive probability of moving
to any path that differs at exactly one interior
coordinate. Transitions between paths differing at two or more interior coordinates are zero. Moreover, $K^g$ is positive semidefinite on $L^2(\widehat{\mathbb V}_{a,b})$.
\end{proposition}

Define
\be \label{del-g} 
\overline g=\max_{i\in\chi}g(i),\qquad
\underline g=\min_{i\in\chi}g(i),\qquad
\Delta_g=\overline g-\underline g.
\ee
As in Theorem \ref{cor:upper-bound-spectral-gap} we can derive a lower bound on the spectral gap of $K^g$. 
\begin{theorem}\label{cor:upper-bound-spectral-gap-g}
The absolute spectral gap $\delta_g$ of $K^g$ satisfies
\be\label{eq:cost-gibbs-gap}
\delta_g\ge\frac{\omega_-^4}{T^2N^3}e^{-\Delta_g}.
\ee
\end{theorem}
The proofs of Proposition~\ref{prop-rev-g} and Theorem~\ref{cor:upper-bound-spectral-gap-g} are given in Section~\ref{sec-pr-additive}. 

We now apply the quantum sampling construction of
Section~\ref{sec-q-gibbs} to the weighted bridge law
$\widehat{\mathbb V}_{a,b}$. The spectral-gap bound
controls one part of the cost. The other depends on
the overlap between the prepared initial state and
the target state.

Write the target coherent state as
\[
    |\widehat{\mathbb V}_{a,b}\rangle
    :=\sum_{x\in\Omega_{a,b}}
      \sqrt{\widehat{\mathbb V}_{a,b}(x)}\,|x\rangle.
\]
Assume that a unitary $A_g$ and its inverse are
available, with
\be \label{gfd} 
    A_g|0\rangle=|r_g\rangle
    :=\sum_{x\in\Omega_{a,b}}\sqrt{r_g(x)}\,|x\rangle,
\ee
where $r_g$ is the chosen initial probability
distribution on $\Omega_{a,b}$ and any auxiliary
registers are returned to zero.
Suppose that a lower bound $c_g\in(0,1]$ is known
such that
\be \label{g-over} 
    \bigl|\langle\widehat{\mathbb V}_{a,b}|r_g\rangle\bigr|
    \ge c_g.
\ee

We also assume coherent row preparation for $K^g$,
its inverse, and the controlled operations used in
Section~\ref{sec-q-gibbs}. Let $\delta_{g,0}>0$ be a known lower bound on the
absolute spectral gap $\delta_g$ of the Gibbs kernel
$K^g$ defined in \eqref{eq:cost-gibbs-kernel}.

\begin{theorem}[Quantum sampling with additive costs]
\label{th:complexity-gibbs-sampler-g}
Under the preparation and access assumptions above,
for every $\varepsilon\in(0,1)$ there is a quantum
sampler that signals success with probability bounded
below by a universal constant. Conditional on success,
its path register has measurement distribution $\rho$
satisfying
\[
    d_{\mathrm{TV}}(\rho,\widehat{\mathbb V}_{a,b})
    \le\varepsilon.
\]
The sampler uses
\[
    O\!\left(
        \frac1{c_g\sqrt{\delta_{g,0}}}
        \log\frac2{c_g\varepsilon}
    \right)
\]
calls to the quantum-walk primitives for $K^g$ and
$O(1/c_g)$ calls to $A_g$ and its inverse.
\end{theorem}

The proof is given in Section~\ref{sec-pr-additive}.

Substituting the explicit gap bound
\[
    \delta_{g,0}
    =\frac{\omega_-^4e^{-\Delta_g}}{T^2N^3}
\]
from Theorem~\ref{cor:upper-bound-spectral-gap-g}
into the complexity bound of
Theorem~\ref{th:complexity-gibbs-sampler-g} gives
\begin{equation}\label{comp-gibbs-g}
    O\!\left(
        \frac{TN^{3/2}e^{\Delta_g/2}}{c_g\omega_-^2}
        \log\frac2{c_g\varepsilon}
    \right)
\end{equation}
calls to the quantum-walk primitives.

This bound counts oracle calls. To obtain the total
running time, one must also account for the construction
and execution costs of the circuits that prepare the
initial state and implement coherent Gibbs updates
according to $K^g$.

The complexity bound \eqref{comp-gibbs-g} depends on
the initial overlap $c_g$, as well as on the cost range
$\Delta_g$ in \eqref{del-g}. Although a bound on $\Delta_g$ controls
the spectral gap, it does not ensure that an initial
state prepared without accounting for the running cost
has a sufficiently large overlap with the target.

For example, take $N=2$, $P(i,j)=1/2$,
$g(1)=0$, and $g(2)=\log9$.
For any fixed endpoints, the ordinary bridge has
independent, uniformly distributed interior states.
Under the weighted bridge law, these states remain
independent but have probabilities $(9/10,1/10)$.

For a single interior state, the ordinary and weighted bridge
marginals are respectively $(1/2,1/2)$ and $(9/10,1/10)$.
With the encoding $1\mapsto|0\rangle$ and $2\mapsto|1\rangle$, the corresponding coherent states are
\[
    |u\rangle=\frac{1}{\sqrt{2}}|0\rangle
              +\frac{1}{\sqrt{2}}|1\rangle,
    \qquad
    |w\rangle=\frac{3}{\sqrt{10}}|0\rangle
              +\frac{1}{\sqrt{10}}|1\rangle.
\]
Since the basis states are orthonormal, their inner product is
\[
    \langle u|w\rangle
    =\sqrt{\frac12\frac9{10}}
     +\sqrt{\frac12\frac1{10}}
    =\frac{2}{\sqrt{5}}.
\]
In this example, the $T-1$ interior states are independent under
each bridge law. Consequently, for fixed endpoints $a$ and $b$,
the full coherent encodings factorise as
\[
    |U_{a,b}\rangle
    =|a\rangle\otimes|u\rangle^{\otimes(T-1)}\otimes|b\rangle,
    \qquad
    |W_{a,b}\rangle
    =|a\rangle\otimes|w\rangle^{\otimes(T-1)}\otimes|b\rangle.
\]
Inner products of tensor products multiply, and the identical
endpoint states contribute a factor of one. Thus,
\[
    \langle U_{a,b}|W_{a,b}\rangle
    =\bigl(\langle u|w\rangle\bigr)^{T-1}
    =\left(\frac{2}{\sqrt{5}}\right)^{T-1}.
\]
Hence the overlap decays exponentially in the number of interior
states.  Thus, even though $\Delta_g=\log9$ is independent of
$T$, this overlap decreases exponentially with $T$.
Using the ordinary bridge as the initial state
therefore introduces an exponentially growing
$1/c_g$ factor into \eqref{comp-gibbs-g}.

This example does not exclude a better initialisation
that accounts for the running cost. It shows that
such an initialisation, its overlap with the target,
and its preparation cost require a separate analysis.

\begin{remark}[Classical comparison and initialisation]
\label{rem:cost-classical-comparison}
If the classical Gibbs sampler starts from a fixed path, the
spectral-gap mixing bound shows that
\[
O\!\left(
\frac{T^2N^3e^{\Delta_g}}{\omega_-^4}
\left[
T\log\frac{N}{\omega_-}
+(T-1)\Delta_g
+\log\frac{2}{\varepsilon}
\right]
\right)
\]
updates of $K^g$ suffice to reach total variation error at most
$\varepsilon$. This follows from the lower bound
\[
\min_x\widehat{\mathbb V}_{a,b}(x)
\ge
\left(\frac{\omega_-}{N}\right)^T
e^{-(T-1)\Delta_g},
\]
proved in Section~\ref{sec-pr-additive}. A better initial distribution may reduce the number of updates.
\end{remark}

\section{Proof of Lemma \ref{lem:dynamic-bridge-to-static-bridge}} \label{proof-d-s}
\begin{proof}[Proof of Lemma \ref{lem:dynamic-bridge-to-static-bridge}]
We first separate the endpoint distribution from the conditional
distribution of the remaining path. The finite-feasibility assumption
ensures that both static and dynamic optimisation problems have a finite minimiser. 
Indeed the feasible sets are compact, and relative entropy is lower
semicontinuous and strictly convex on the support of the reference law.

Let $\mathbb Q\in\Lambda(\mu_0,\mu_T)$ satisfy
$\mathbb Q\ll\mathbb P$, and denote its endpoint law by
$\mathbb Q_{0,T}$ (see \eqref{ref-stat}).  For an endpoint pair $(a,b)$ with
$\mathbb Q_{0,T}(a,b)>0$, absolute continuity implies
$\mathbb P_{0,T}(a,b)>0$, so both conditional laws on the starting and ending points below are defined.  The
chain rule for relative entropy gives
\be \label{kl-comp}
\begin{aligned}
\operatorname{KL}(\mathbb Q\,\|\,\mathbb P)
&=\operatorname{KL}(\mathbb Q_{0,T}\,\|\,\mathbb P_{0,T})\\
&\quad+\sum_{\substack{a,b\in\chi:\\
                 \mathbb Q_{0,T}(a,b)>0}}
\mathbb Q_{0,T}(a,b)
\operatorname{KL}\!\left(
\mathbb Q(\,\cdot\mid X_0=a,X_T=b)
\,\middle\|\,
\mathbb P(\,\cdot\mid X_0=a,X_T=b)
\right).
\end{aligned}
\ee
Each of the conditional KL terms in \eqref{kl-comp} is nonnegative and equals zero
exactly when the two conditional laws agree.  Conversely, any endpoint
coupling $q\in\Pi(\mu_0,\mu_T)$ with $q\ll\mathbb P_{0,T}$ can be lifted to
the path law as follows: 
\be \label{bb1}
\mathbb Q^q(x)=q(x_0,x_T)\mathbb P(x\mid X_0=x_0,X_T=x_T),
\ee
with the product interpreted as zero on endpoint pairs outside the support of
$\mathbb P_{0,T}$.  From \eqref{kl-comp} and \eqref{bb1} it follows that minimising the dynamic KL divergence in \eqref{dyn-br}
 is
equivalent to minimising its endpoint term.  Hence the endpoint law of
$\mathbb Q^*$ is the static bridge $\mathbb Q^*_{0,T}$ with
marginals $(\mu_0,\mu_T)$ and reference law $\mathbb P_{0,T}$. For every
endpoint pair sampled by that law we have 
\[
\mathbb Q^*(\,\cdot\mid X_0=a,X_T=b)
=\mathbb P(\,\cdot\mid X_0=a,X_T=b).
\]
This proves (i). Under
$\mathbb P_{0,T}\sim\mu_0\otimes\mu_T$,
Proposition~\ref{prop:static-potentials} gives the factorisation in
 (ii) on the product of the marginal supports. To interpret
the formula for all endpoint pairs, extend each potential by $1$
outside its marginal support; the reference endpoint mass is zero
there. Multiplying the endpoint factorisation by the reference
conditional law proves (iii).
\end{proof}
 
\section{Proofs of Section \ref{sec:setup-gibbs-sampler}} \label{sec-pf-gibbs}

\begin{proof}[Proof of Proposition \ref{prop-rev}]
First note that \eqref{eq:definition-Gibbs-sampler} is an average of
probability kernels and is therefore stochastic. For irreducibility, let $x,y\in\Omega_{a,b}$ and define, for
$0\le m\le T-1$,
\[
z^{(m)}=(a,y_1,\dots,y_m,x_{m+1},\dots,x_{T-1},b).
\]
Then $z^{(0)}=x$, $z^{(T-1)}=y$, and each consecutive pair differs only at
coordinate $m$.  Strict positivity of $P$ and
\eqref{eq:conditionned-interior-proba-P} give
$K(z^{(m-1)},z^{(m)})>0$ for every $1\le m\le T-1$.  Thus
$K^{T-1}(x,y)>0$.  From the assumption $P(i,j)>0$ for all $i,j\in\chi$ it follows that every conditional probability in \eqref{eq:conditionned-interior-proba-P} assigns positive probability
to the current coordinate, so $K(x,x)>0$ for every $x$; hence the chain is
aperiodic.  This proves (i).

For (ii), the detailed balance condition \eqref{dbc} is immediate when $x=y$.
If $x$ and $y$ differ at more than one coordinate, both transition
probabilities are zero. Otherwise they differ at exactly one
coordinate $t$. If this is the case, write $x_{-t}$ for all coordinates of $x$ except $x_t$,
and let
\[
\pi_{-t}(x_{-t})
=\mathbb P_{a,b}(X_{-t}=x_{-t})
\]
be the probability of these shared coordinates.
By \eqref{eq:conditionned-interior-proba-P} and the definition of conditional probability,
\[
\mathbb P_{a,b}(x)
=\pi_{-t}(x_{-t})c_t(x_t\mid x_{-t}).
\]
To move from $x$ to $y$, the Gibbs sampler must select
coordinate $t$, with probability $1/(T-1)$, and then
draw the value $y_t$. Thus by \eqref{eq:definition-Gibbs-sampler}, 
\[
K(x,y)=\frac{1}{T-1}c_t(y_t\mid x_{-t}).
\]
Consequently,
\[
\mathbb P_{a,b}(x)K(x,y)
=\frac{\pi_{-t}(x_{-t})}{T-1}
c_t(x_t\mid x_{-t})c_t(y_t\mid x_{-t}).
\]
Since $x_{-t}=y_{-t}$, exchanging $x$ and $y$ simply
exchanges the two conditional-probability factors.
Hence
\[
\mathbb P_{a,b}(x)K(x,y)
=\mathbb P_{a,b}(y)K(y,x).
\]
Together with the cases considered above, this proves
detailed balance and therefore reversibility with
respect to $\mathbb P_{a,b}$.

(iii) For each $t\in\{1,\ldots,T-1\}$, define the
Markov kernel $K_t$ on $\Omega_{a,b}$ by
\[
K_t(x,y)
=\mathds1_{\{x_{-t}=y_{-t}\}}\,
c_t(y_t\mid x_{-t}),
\qquad x,y\in\Omega_{a,b}.
\]
This kernel resamples coordinate $t$ from its conditional
distribution, leaving all other coordinates unchanged.
The Gibbs kernel in \eqref{eq:definition-Gibbs-sampler}
is therefore
\begin{equation}\label{k-pf}
K=\frac1{T-1}\sum_{t=1}^{T-1}K_t.
\end{equation}
For any function $f:\Omega_{a,b}\to\mathbb C$,
\[
(K_tf)(x)
=\sum_{y\in\Omega_{a,b}}K_t(x,y)f(y)
=\mathbb E_{\mathbb P_{a,b}}
\bigl[f(X)\mid X_{-t}=x_{-t}\bigr].
\]
Since $K_tf$ depends only on $X_{-t}$, averaging it
again over coordinate $t$ leaves it unchanged.
Thus $K_t^2f=K_tf$, and
\[
\mathbb E_{\mathbb P_{a,b}}[f-K_tf\mid X_{-t}]=0.
\]
The tower property, together with the fact that $K_tf$
can be taken outside the conditional expectation, gives
\[
\begin{aligned}
\langle f-K_tf,K_tf\rangle_{\mathbb P_{a,b}}
&=\mathbb E_{\mathbb P_{a,b}}\!\left[
    \overline{f-K_tf}\,K_tf\right]\\
&=\mathbb E_{\mathbb P_{a,b}}\!\left[
    \overline{
      \mathbb E_{\mathbb P_{a,b}}[f-K_tf\mid X_{-t}]
    }\,K_tf\right]\\
&=0.
\end{aligned}
\]
Consequently,
\[
\langle f,K_tf\rangle_{\mathbb P_{a,b}}
=\langle K_tf,K_tf\rangle_{\mathbb P_{a,b}}
=\|K_tf\|_{L^2(\mathbb P_{a,b})}^2\ge0.
\]
Combining this identity with \eqref{k-pf} yields
\begin{equation}\label{k-l-bnd}
\langle f,Kf\rangle_{\mathbb P_{a,b}}
=\frac1{T-1}\sum_{t=1}^{T-1}
\|K_tf\|_{L^2(\mathbb P_{a,b})}^2
\ge0.
\end{equation}
Therefore $K$ is positive semidefinite.
\end{proof}

We next bound the spectral gap by comparing changes in a function
along single-coordinate updates. Write $\pi=\mathbb P_{a,b}$.
For $f:\Omega_{a,b}\to\mathbb C$, the Dirichlet form and variance are
\[
\mathcal E(f)=\langle f,(I-K)f\rangle_\pi
=\frac12\sum_{x,y\in\Omega_{a,b}}|f(x)-f(y)|^2\pi(x)K(x,y)
\]
and
\[
\operatorname{Var}_\pi(f)
=\frac12\sum_{x,y\in\Omega_{a,b}}|f(x)-f(y)|^2\pi(x)\pi(y).
\]
Let $\lambda_2$ be the second largest eigenvalue of $K$ (see \eqref{eigen-ord}). Since $K$ is reversible and irreducible, the variational characterisation of the spectral gap \cite[Lemma~13.7 and Remark~13.8]{levin2017markov} gives
\[
1-\lambda_2=\min_{f\ \mathrm{nonconstant}}
\frac{\mathcal E(f)}{\operatorname{Var}_\pi(f)}.
\]
Here minimisation is done over nonconstant functions. 

We now specify a canonical path between every ordered pair
$u,v\in\Omega_{a,b}$. For $0\leq i\leq T-1$, set
\be \label{lead-p}
w_i^{u\to v}=(a,v_1,\dots,v_i,u_{i+1},\dots,u_{T-1},b).
\ee

Thus $w_0^{u\to v}=u$ and $w_{T-1}^{u\to v}=v$. Omit any step
whose two consecutive paths coincide, and define the collection of
directed edges
\be \label{gam}
\Gamma_{u,v}=
\left\{(w_{i-1}^{u\to v},w_i^{u\to v}):
1\leq i\leq T-1,\ w_{i-1}^{u\to v}\neq w_i^{u\to v}\right\}.
\ee
Every edge in this path changes one coordinate and has positive transition
probability, and $|\Gamma_{u,v}|\leq T-1$.

Let
\[
\mathcal E_K=
\left\{(x,y)\in\Omega_{a,b}^2:x\neq y,\ K(x,y)>0\right\}.
\]
By construction, $\Gamma_{u,v}\subseteq\mathcal E_K$ for every
$u,v\in\Omega_{a,b}$.
For $(x,y)\in\mathcal E_K$, define its canonical-path congestion by
\be \label{rhoxy}
\rho(x,y)=\frac{1}{\pi(x)K(x,y)}
\sum_{u,v:\,(x,y)\in\Gamma_{u,v}}
\pi(u)\pi(v)|\Gamma_{u,v}|,
\ee
where the sum is over $u,v\in\Omega_{a,b}$, and set
\be \label{rho-def}
\rho=\max_{(x,y)\in\mathcal E_K}\rho(x,y).
\ee
Diagonal pairs are excluded because steps that leave the state unchanged
have been omitted. For every $(x,y)\in\mathcal E_K$, we have
$\pi(x)>0$ and $K(x,y)>0$, so the denominator is strictly positive.

The following proposition is the canonical-path inequality; see Corollary~13.21 of
\cite{levin2017markov}.
\begin{proposition}\label{prop:spectral-gap-notabsolute-to-absolute}
Under the same assumptions as Proposition \ref{prop-rev}, let $\lambda_2$ be the second-largest eigenvalue of $K$.  Then
\[
\frac1{1-\lambda_2}\le\rho.
\]
\end{proposition}

\begin{theorem}\label{th:upper-bound-rho}
Assume that \eqref{om} holds. Then we have 
\[
\rho\le\frac{T^2N^3}{\omega_-^4}.
\]
\end{theorem}

\begin{proof}
Fix an off-diagonal edge $(x,y)\in\mathcal E_K$, hence there is a unique
$t\in\{1,\dots,T-1\}$ for which $x_{-t}=y_{-t}$.  Let
\[
C_{x,y}=\{(u,v)\in\Omega_{a,b}^2:(x,y)\in\Gamma_{u,v}\}.
\]
The definition of the canonical paths \eqref{gam} gives
\be \label{free-co}
\begin{aligned}
C_{x,y}=\bigl\{\bigl(&
(a,u_1,\dots,u_{t-1},x_t,x_{t+1},\dots,x_{T-1},b),\\
&(a,x_1,\dots,x_{t-1},y_t,v_{t+1},\dots,v_{T-1},b)
\bigr):\\
&u_1,\dots,u_{t-1},v_{t+1},\dots,v_{T-1}\in\chi\bigr\}.
\end{aligned}
\ee
At $t=1$ or $t=T-1$, the corresponding list of free coordinates is empty.
Moreover from \eqref{eq:definition-Gibbs-sampler} we have 
\be \label{k-eq}
K(x,y)=\frac{c_t(y_t\mid x_{-t})}{T-1}.
\ee
Since every canonical path has length at most $T-1$, equations
\eqref{rhoxy}, \eqref{free-co}, and \eqref{k-eq} and the notation $\mathbb P_{a,b} = \pi$ imply
\be \label{g1}
\rho(x,y)\leq
\frac{(T-1)^2 I_1(x,y)I_2(x,y)}
{\mathbb P_{a,b}(x)c_t(y_t\mid x_{-t})},
\ee 
where
\be \label{i-1} 
I_1(x,y)
=
\sum_{(u_1,\dots,u_{t-1})\in\chi^{t-1}}
\mathbb P_{a,b}
\bigl(a,u_1,\dots,u_{t-1},x_t,x_{t+1},\dots,x_{T-1},b\bigr)
\ee
and
\be \label{i-2} 
I_2(x,y)
=
\sum_{(v_{t+1},\dots,v_{T-1})\in\chi^{T-t-1}}
\mathbb P_{a,b}
\bigl(a,x_1,\dots,x_{t-1},y_t,v_{t+1},\dots,v_{T-1},b\bigr).
\ee
At $t=1$ or $t=T-1$, the corresponding tuple of free coordinates
is empty, and its sum consists of a single term.
The free coordinates in the two paths can be chosen independently,
so the sum of $\pi(u)\pi(v)$ over $C_{x,y}$ factors as
$I_1(x,y)I_2(x,y)$.

Set $Z=P^T(a,b)$, so that
$\mathbb P(X_0=a,X_T=b)=\mu_0(a)Z$.  Introduce the forward and backward
quantities
\be \label{l-r} 
L_t(r)=P^t(a,r),\qquad R_t(r)=P^{T-t}(r,b),
\ee
and the fixed prefix and suffix products
\[
A_t(x)=\prod_{i=0}^{t-2}P(x_i,x_{i+1}),\qquad
B_t(x)=\prod_{i=t}^{T-1}P(x_i,x_{i+1}),
\]
where empty products equal one. From \eqref{i-1} we get 
\be \label{i1}
I_1(x,y)=\frac{L_t(x_t)B_t(x)}{Z}.
\ee
From \eqref{i-2} we have 
\be \label{i2}
I_2(x,y)=\frac{A_t(x)P(x_{t-1},y_t)R_t(y_t)}{Z}.
\ee
We can also write,
\be \label{i3}
I_3(x):=\mathbb P_{a,b}(x)
=\frac{A_t(x)P(x_{t-1},x_t)B_t(x)}{Z},
\ee
and using \eqref{eq:conditionned-interior-proba-P}, 
\be \label{i4}
I_4(x,y):=c_t(y_t\mid x_{-t})
=\frac{P(x_{t-1},y_t)P(y_t,x_{t+1})}{D_t(x)},
\qquad
D_t(x)=\sum_{z\in\chi}P(x_{t-1},z)P(z,x_{t+1}).
\ee
Combining \eqref{g1}--\eqref{i4} yields
\be \label{rho-i}
\rho(x,y)\le (T-1)^2
\frac{I_1(x,y)I_2(x,y)}{I_3(x)I_4(x,y)}.
\ee
After cancelling the common prefix and suffix products,
\be \label{f2}
\rho(x,y)\le\frac{(T-1)^2}{I_4(x,y)}
\frac{P(x_{t-1},y_t)}{P(x_{t-1},x_t)}
\frac{L_t(x_t)R_t(y_t)}{Z}.
\ee
The Chapman--Kolmogorov identity gives
\be \label{f3}
Z=P^T(a,b)=\sum_{z_t\in\chi}L_t(z_t)R_t(z_t).
\ee
To bound the terms in \eqref{f2}, use the elementary inequality
\be \label{inff} 
\frac{\sum_k a_k u_k}{\sum_k a_k v_k}
\ge\min_k\frac{u_k}{v_k},
\qquad a_k\ge0,\quad \sum_k a_k>0,\quad u_k,v_k>0.
\ee
From \eqref{l-r} and \eqref{inff} we get 
\[
\begin{aligned}
\frac{L_t(z_t)}{L_t(x_t)}
&=\frac{\sum_rP^{t-1}(a,r)P(r,z_t)}
        {\sum_rP^{t-1}(a,r)P(r,x_t)}
\ge\min_r\frac{P(r,z_t)}{P(r,x_t)},\\
\frac{R_t(z_t)}{R_t(y_t)}
&=\frac{\sum_sP(z_t,s)P^{T-t-1}(s,b)}
        {\sum_sP(y_t,s)P^{T-t-1}(s,b)}
\ge\min_s\frac{P(z_t,s)}{P(y_t,s)}.
\end{aligned}
\]
Multiplying these inequalities yields
\be \label{f4}
\frac{L_t(z_t)}{L_t(x_t)}
\frac{R_t(z_t)}{R_t(y_t)}
\ge
\min_{r,s\in\chi}
\frac{P(r,z_t)P(z_t,s)}{P(r,x_t)P(y_t,s)}.
\ee
These formulas also cover $t=1$ and $t=T-1$, with $P^0=I$;
the corresponding sum then fixes the boundary state.
Combining \eqref{f2}--\eqref{f4}, we obtain
\be \label{f5}
\rho(x,y)\le\frac{(T-1)^2}{I_4(x,y)}
\frac{P(x_{t-1},y_t)}{P(x_{t-1},x_t)}
\left(
\sum_{z_t\in\chi}\min_{r,s\in\chi}
\frac{P(r,z_t)P(z_t,s)}{P(r,x_t)P(y_t,s)}
\right)^{-1}.
\ee
Substituting \eqref{i4} cancels the factor $P(x_{t-1},y_t)$ and gives
\[
\rho(x,y)\le
\frac{(T-1)^2D_t(x)}
{P(x_{t-1},x_t)P(y_t,x_{t+1})}
\left(
\sum_{z_t\in\chi}\min_{r,s\in\chi}
\frac{P(r,z_t)P(z_t,s)}{P(r,x_t)P(y_t,s)}
\right)^{-1}.
\]
Note that $D_t(x)\le1$.  Assumption \eqref{om} gives
\[
P(x_{t-1},x_t)P(y_t,x_{t+1})\ge\frac{\omega_-^2}{N^2}
\]
and
\[
\sum_{z_t\in\chi}\min_{r,s\in\chi}
\frac{P(r,z_t)P(z_t,s)}{P(r,x_t)P(y_t,s)}
\ge N\frac{\omega_-^2}{N^2}=\frac{\omega_-^2}{N}.
\]
It follows that
\[
\rho(x,y)\le\frac{(T-1)^2N^3}{\omega_-^4}
\le\frac{T^2N^3}{\omega_-^4}.
\]
Taking the maximum over $(x,y)\in\mathcal E_K$ and using \eqref{rho-def} proves the theorem.
\end{proof}

Next we prove the spectral-gap theorem.
\begin{proof}[Proof of Theorem \ref{cor:upper-bound-spectral-gap}]
By Proposition~\ref{prop-rev}(iii), all eigenvalues of $K$ are nonnegative;
in particular,
\be \label{lam-n-lbnd}
\lambda_{N^{T-1}}\ge0.
\ee
Theorem~\ref{th:upper-bound-rho} and
Proposition~\ref{prop:spectral-gap-notabsolute-to-absolute} give
\be \label{lam-2-lbnd}
1-\lambda_2\ge\frac{\omega_-^4}{T^2N^3}.
\ee
Since the spectrum is nonnegative, \eqref{sp-rev} reduces to
$\delta=1-\lambda_2$.  The result follows from \eqref{lam-2-lbnd}.
\end{proof}

\section{Proof of Lemma \ref{lem-ol}}\label{sec-pf-qsamp}
 \begin{proof}
Since $\sum_{j\in\chi}P(i,j)=1$ for every $i\in\chi$, summing successively over $x_{T-1},\dots,x_1$ gives $\sum_{x\in\Omega_{a,b}}q_{a,b}(x)=1$, as the product in \eqref{sigma} contains no transition to the fixed endpoint $b$.

Set $Z=P^T(a,b)$. The assumption of Lemma~\ref{lem-ol} gives
\[
    Z=\sum_{i\in\chi}P^{T-1}(a,i)P(i,b)
    \ge\eta>0.
\]
For every path $x$ with $q_{a,b}(x)>0$, the state
$x_{T-1}$ is reachable from $a$ at $T-1$ steps, so again by assumption 
$P(x_{T-1},b)\ge\eta$. Together with  \eqref{sigma} and the fact that $Z\le1$ we get 
\[
    \pi(x)
    =\frac{q_{a,b}(x)P(x_{T-1},b)}{Z}
    \ge\eta q_{a,b}(x).
\]
This inequality also holds when $q_{a,b}(x)=0$.
Consequently,
\be \label{eq:initial-domination}
    q_{a,b}(x)\le\eta^{-1}\pi(x), \quad \textrm{for all }
     x\in\Omega_{a,b}.
\ee
From \eqref{sigma}, \eqref{over}, and \eqref{eq:initial-domination} we get 
\[
    \theta_{a,b}
    =\sum_{x\in\Omega_{a,b}}\sqrt{\pi(x)q_{a,b}(x)}
    \ge\sqrt{\eta}\sum_{x\in\Omega_{a,b}}q_{a,b}(x) 
    =\sqrt{\eta},
\]
which completes the proof. 
\end{proof}

 \section{Proof of Proposition \ref{prop:amplitude-amplification}}
\label{sec:amplitude-amplification}
We first provide the formulation of amplitude amplification following
\cite{brassardamplitudeamplification}.
Recall the decomposition of the state $\Psi$ in \eqref{psi-dec}. Let $s=\|\Psi_*\|^2$ be the success probability, and define the
amplification angle $\vartheta$ by
\begin{equation}\label{sin-th}
\sin^2\vartheta=s,\qquad 0\le\vartheta\le\frac\pi2.
\end{equation}
 For $0<s<1$, the Grover iterate \eqref{r0} acts on the two
components in \eqref{psi-dec} as
\[
Q|\Psi_*\rangle=(1-2s)|\Psi_*\rangle-2s|\Psi_\perp\rangle,
\qquad
Q|\Psi_\perp\rangle
=2(1-s)|\Psi_*\rangle+(1-2s)|\Psi_\perp\rangle.
\]
Indeed, $Q=(2|\Psi\rangle\langle\Psi|-I)(I-2\Pi_*)$, where
$(I-2\Pi_*)|\Psi_*\rangle=-|\Psi_*\rangle$ and
$(I-2\Pi_*)|\Psi_\perp\rangle=|\Psi_\perp\rangle$.
By orthogonality, $\langle\Psi|\Psi_*\rangle=s$ and
$\langle\Psi|\Psi_\perp\rangle=1-s$; applying the first factor
and substituting $|\Psi\rangle=|\Psi_*\rangle+|\Psi_\perp\rangle$
gives the two identities.

The corresponding normalised eigenvectors of the Grover iterate $Q$
are given below; they satisfy
$Q|\Psi_\pm\rangle=e^{\pm 2i\vartheta}|\Psi_\pm\rangle$,
where $\vartheta=\arcsin\sqrt{s}$:
\begin{equation}\label{eq:eigenbase-amplitude-amplification}
|\Psi_\pm\rangle
=\frac1{\sqrt2}
\left(\frac{|\Psi_*\rangle}{\sqrt s}
\pm i\,\frac{|\Psi_\perp\rangle}{\sqrt{1-s}}\right).
\end{equation}
Consequently,
\begin{equation}\label{eq:iteration-operator-Q-amplitude-amplification}
Q^n|\Psi\rangle
=\frac{\sin((2n+1)\vartheta)}{\sqrt s}|\Psi_*\rangle
+\frac{\cos((2n+1)\vartheta)}{\sqrt{1-s}}|\Psi_\perp\rangle.
\end{equation}
For $s=1$, the input is already entirely marked and $Q$ changes
only its global phase. The case $s=0$ is excluded when a positive
lower bound is supplied.

\paragraph{Estimating the amplification angle.}
Let $\mathcal H$ be the Hilbert space of all registers on which
the Grover iterate $Q$ in \eqref{r0} acts.
For an integer $M\geq2$, introduce a phase register $\mathbb C^M$,
with basis $|j\rangle$, $0\leq j<M$. Define
\[
\Lambda_M(Q)|j\rangle|y\rangle=|j\rangle Q^j|y\rangle,
\qquad |y\rangle\in\mathcal H,
\]
and
\[
F_M|j\rangle=\frac{1}{\sqrt M}
\sum_{k=0}^{M-1}e^{2\pi ijk/M}|k\rangle.
\]
Thus, $\Lambda_M(Q)$ acts on $\mathbb C^M\otimes\mathcal H$,
whereas $F_M$ acts only on the phase register.

The following algorithm estimates the norm of the successful
component $|\Psi_*\rangle$ in the initial state
$|\Psi\rangle=\mathcal A|0\rangle$.
It returns an estimate $\widetilde{\vartheta}$ of
$\vartheta=\arcsin\|\Psi_*\|$, so that
$\sin\widetilde{\vartheta}$ estimates $\|\Psi_*\|$ and
$\sin^2\widetilde{\vartheta}$ estimates the initial success
probability $s$.

Specifically, the algorithm prepares
$M^{-1/2}\sum_{j=0}^{M-1}|j\rangle|\Psi\rangle$,
applies $\Lambda_M(Q)$, and then applies $F_M^{-1}$ to the
phase register. It measures that register in the basis
$\{|k\rangle:0\leq k<M\}$ and records the resulting integer $k$.
For an eigenvector of $Q$ with eigenvalue $e^{2\pi i\phi}$,
these operations concentrate the measurement probabilities
near values of $k$ for which $k/M$ approximates $\phi$ modulo one.

The input $|\Psi\rangle$ is a superposition of the eigenvectors
in \eqref{eq:eigenbase-amplitude-amplification}, whose phases
are $\vartheta/\pi$ and $1-\vartheta/\pi$.
The algorithm therefore uses the distance of $k/M$ from the
nearest integer, namely $\min\{k/M,1-k/M\}$.
Both exact phases have distance $\vartheta/\pi$ from the
nearest integer, so multiplying this distance by $\pi$ gives
the angle estimate
\[
\widetilde{\vartheta}=\pi\min\{k/M,1-k/M\}.
\]

\begin{algorithm}[H]
\caption{Folded amplitude estimation
\cite[Section~4]{brassardamplitudeamplification}}
\label{alg:cap}
\begin{algorithmic}
\Require Unitary preparation circuit $\mathcal A$ satisfying
$\mathcal A|0\rangle=|\Psi\rangle$, controlled Grover iterate
$Q$ defined in \eqref{r0}, and integer $M\geq2$.
\State Prepare $|0\rangle\,\mathcal A|0\rangle$.
\State Apply $F_M$ to the first register.
\State Apply $\Lambda_M(Q)$.
\State Apply $F_M^{-1}$ to the first register.
\State Measure the first register, obtaining $y\in\{0,\ldots,M-1\}$.
\State Return $\widetilde\vartheta=\pi\min\{y,M-y\}/M$.
\end{algorithmic}
\end{algorithm}

\begin{theorem}[{Consequence of amplitude estimation
\cite[Theorems~11--12]{brassardamplitudeamplification}}]
\label{th:amplitude-estimation}
Algorithm~\ref{alg:cap} returns
$\widetilde\vartheta\in[0,\pi/2]$ satisfying
\[
|\widetilde\vartheta-\vartheta|\le\frac\pi M
\]
with probability at least $8/\pi^2$. For every integer $k\ge2$,
the error is at most $k\pi/M$ with probability at least
$1-1/(2(k-1))$. The algorithm uses $O(M)$ controlled applications
of $Q$ and one preparation of $|\Psi\rangle$.
These statements include $s=0$ and $s=1$.
\end{theorem}
To see why the phase-estimation bounds imply this statement, let
$d_{\mathrm{circ}}(u,v)=\min_{k\in\mathbb Z}|u-v-k|$ denote distance
modulo one. The folding map is
$F(u)=\pi d_{\mathrm{circ}}(u,0)$, and the triangle inequality gives
\[
|F(u)-F(v)|\le\pi d_{\mathrm{circ}}(u,v).
\]
For either eigenphase $v=\vartheta/\pi$ or
$v=1-\vartheta/\pi$, we have $F(v)=\vartheta$.
Thus the circular phase-error bounds in
\cite[Theorem~11]{brassardamplitudeamplification} give the displayed
angle-error bounds, including when the input is a superposition
of the two eigenvectors. The probability estimate
$\widetilde s=\sin^2(\pi y/M)$ is unchanged by folding.

\begin{proof}[Proof of Proposition~\ref{prop:amplitude-amplification}]
Choose
\begin{equation}\label{m-k}
M=\left\lceil\frac{40}{c}\right\rceil.
\end{equation}
The phase estimate is used to choose an iteration count that
moves the marked component close to its maximum. If the estimate
already indicates a sufficiently large marked component, no
amplification is needed. For the initial state $|\Psi\rangle$, run Algorithm~\ref{alg:cap} and return
\be \label{n-choice}
n=
\begin{cases}
0,& \textrm{if }\widetilde\vartheta<c/2
       \ \text{or}\ \widetilde\vartheta\ge\pi/6,\\[2mm]
\left\lfloor\pi/(4\widetilde\vartheta)\right\rfloor, 
  & \textrm{if } c/2\le\widetilde\vartheta<\pi/6.
\end{cases}
\ee
This rule is defined on every measurement outcome, including zero,
and always gives 
\be \label{n-c} 
n\le \frac{\pi}{2c}.
\ee
Let $\mathcal E$ denote the estimation event
\begin{equation}\label{eve-s}
|\widetilde\vartheta-\vartheta|\le\frac\pi M.
\end{equation}
By Theorem \ref{th:amplitude-estimation} it has probability at least $8/\pi^2$.
Using the assumption $\|\Psi_*\|\ge c$ and $\vartheta \in [0,\pi/2]$ we have $\vartheta\ge\sin\vartheta=\|\Psi_*\|\ge c$. Hence
on $\mathcal E$,
\[
\widetilde\vartheta\ge c-\frac\pi M
\ge c\left(1-\frac\pi{40}\right)>\frac c2.
\]
If $\widetilde\vartheta\ge\pi/6$, then $n=0$ and from \eqref{n-c} and \eqref{eve-s} we have 
\[
\|\Psi_*\|=\sin\vartheta
\ge\sin\widetilde\vartheta-\frac\pi M
\ge\frac12-\frac\pi{40}>\frac25.
\]
This also covers the case where $s=1$ on $\mathcal E$.

If $0<\widetilde\vartheta<\pi/6$, the rounding rule in \eqref{n-choice} gives
\begin{equation}\label{d1}
\left|(2n+1)\widetilde\vartheta-\frac\pi2\right|
\le\widetilde\vartheta,
\end{equation}
Hence by the triangle inequality, \eqref{d1} and then application of \eqref{n-c} and \eqref{eve-s} we have on $\mathcal E$,
\begin{equation}\label{d22}
\begin{aligned}
\left|(2n+1)\vartheta-\frac\pi2\right|
&\le\widetilde\vartheta
 +(2n+1)|\vartheta-\widetilde\vartheta|\\
&\le\frac\pi6+\frac{\pi^2}{Mc}+\frac\pi M\\
& <\frac\pi3, 
\end{aligned}
\end{equation}
where we have used \eqref{m-k} in the last inequality. 

Having chosen $n$ by amplitude estimation (see Algorithm \ref{alg:cap}), we apply $Q$ exactly
$n$ times on $|\Psi\rangle$ and get the state in 
\eqref{eq:iteration-operator-Q-amplitude-amplification}. 
Recall that from \eqref{psi-dec}, $\Pi_*|\Psi_*\rangle=|\Psi_*\rangle$,  
$\Pi_*|\Psi_\perp\rangle=0$. Together with $\|\Psi_*\|=\sqrt{s}$ we get
\[
\|\Pi_*Q^n|\Psi\rangle\|
=|\sin((2n+1)\vartheta)|> \frac25,
\]
where the strict inequality follows from \eqref{d22}.
The estimation stage uses $O(M)=O(1/c)$ controlled Grover calls,
and the returned $n$ satisfies the same bound on all outcomes.

Thus, on $\mathcal E$, the returned integer $n$ satisfies
$\|\Pi_*Q^n|\Psi\rangle\|\geq2/5$.
Since $\mathbb P(\mathcal E)\geq8/\pi^2$, this proves the
claimed probability bound.
The selection rule gives $n=O(1/c)$ for every outcome.
Taking $M=\lceil40/c\rceil$, the phase estimation algorithm uses
one preparation of $|\Psi\rangle$ and $O(M)=O(1/c)$
controlled applications of $Q$.

\end{proof}

\section{Proof of Proposition \ref{prop-poly}}
\label{sec:reflection-eigenspace}
We first construct a polynomial $\Gamma$ satisfying $\Gamma(1)=1$
and having small modulus at every other eigenvalue of $\mathcal W$.
Thus, $\Gamma(\mathcal W)$ approximates the projector $\Pi$ onto
the eigenspace of $\mathcal W$ with eigenvalue $1$.
We then implement this polynomial as the success branch of a
unitary circuit.

\begin{theorem}[{\cite[Theorem~1]{baptistereflectioneigenspace}}]
\label{th:PUE-reflection-eigenspace}
Let $U$ be unitary, and let $\Gamma$ be a polynomial of degree at
most $d$ satisfying
\begin{equation}\label{gam-cond}
\sup_{|z|=1}|\Gamma(z)|\le1.
\end{equation}
There exists a polynomial $\Phi$ of degree at most $d$ satisfying
$|\Gamma(z)|^2+|\Phi(z)|^2=1$ on the unit circle, and a unitary
$V_\Gamma$ such that, for every $|\psi\rangle$,
\[
V_\Gamma(|0\rangle|\psi\rangle)
=|0\rangle\Gamma(U)|\psi\rangle
+|1\rangle\Phi(U)|\psi\rangle.
\]
The construction uses at most $d$ controlled-$U$ calls, $O(d+1)$
additional single-qubit gates, and one ancillary qubit.
In particular,
$(\langle0|\otimes I)V_\Gamma(|0\rangle\otimes I)=\Gamma(U)$
is an ordinary projected unitary encoding.
\end{theorem}

\begin{proof}[Proof of Proposition~\ref{prop-poly}]
Let $\Pi$ be the projector of $\mathcal W$ for
eigenvalue $1$ defined in \eqref{proj}, and take all eigenphases in $[-\pi,\pi]$.
By assumption, their nonzero magnitudes are at least
$\Delta_0\in(0,\pi]$.
For any integers $t,n\ge1$, set
\begin{equation}\label{eq:filter-polynomial}
\Gamma_{t,n}(z)
=\left(\frac1t\sum_{j=0}^{t-1}z^j\right)^n.
\end{equation}
For $|z|\le1$, this polynomial has modulus at most one, and
$\Gamma_{t,n}(1)=1$.  For a nonzero eigenphase $\lambda$, the geometric-sum formula gives
\be \label{dd1}
\left|\frac1t\sum_{j=0}^{t-1}e^{ij\lambda}\right|
=
\frac{|1-e^{it\lambda}|}{t|1-e^{i\lambda}|}
\leq\frac{2}{t|1-e^{i\lambda}|}.
\ee
Since $\Delta_0\leq|\lambda|\leq\pi$ and
$|1-e^{i\lambda}|=2\sin(|\lambda|/2)$, we have
\be \label{dd2}
|1-e^{i\lambda}|\geq|1-e^{i\Delta_0}|.
\ee
Choose
\be \label{dd3}
t=\left\lceil\frac{2e}{|1-e^{i\Delta_0}|}\right\rceil,
\qquad
n=\left\lceil\ln\frac1\eta\right\rceil.
\ee
By the choice of $t$, we have
$t|1-e^{i\Delta_0}|\geq2e$, so together with \eqref{dd1}--\eqref{dd3} we get 
\[
|\Gamma_{t,n}(e^{i\lambda})|
=
\left|\frac1t\sum_{j=0}^{t-1}e^{ij\lambda}\right|^n
\leq e^{-n}\leq\eta.
\]
On the eigenspace with eigenvalue $1$, we have
$\Gamma_{t,n}(1)=1$, so $\Gamma_{t,n}(\mathcal W)-\Pi$
vanishes. On every other eigenspace, $\Pi$ vanishes and
$\Gamma_{t,n}(\mathcal W)$ acts by a scalar of modulus at most
$\eta$. Since $\mathcal W$ is unitary, its eigenspaces are
orthogonal, and hence
\[
\|\Gamma_{t,n}(\mathcal W)-\Pi\|\leq\eta.
\]

Since $|e^{i\Delta_0}-1|=2\sin(\Delta_0/2)
\ge2\Delta_0/\pi$, we have by \eqref{dd3} $t=O(1/\Delta_0)$.
Thus the degree of $\Gamma_{t,n}$ is
\[
(t-1)n=O\!\left(\frac1{\Delta_0}\log\frac2\eta\right).
\]
Applying Theorem~\ref{th:PUE-reflection-eigenspace} gives the asserted
unitary encoding and query bound.
\end{proof}

\section{Proof of Theorem \ref{th:complexity-gibbs-sampler}}
\label{pf-gibbs}
\begin{proof}
We use the preparation $|\Sigma_{a,b}\rangle=A_{a,b}|0\rangle$ from \eqref{sigma} and set
\be \label{c-dl-0} 
c=\sqrt{\frac{\omega_-}{N}},
\qquad
\delta_0=\frac{\omega_-^4}{T^2N^3}.
\ee
Lemma~\ref{lem-ol} and Assumption~\ref{assump-spec} give $\theta_{a,b}\geq c$ for the overlap
defined in \eqref{over}, and
Theorem~\ref{cor:upper-bound-spectral-gap} gives
$\delta\geq\delta_0$ for the absolute spectral gap of $K$. The proof of the result is divided into the following steps.
 
\paragraph{Approximation of the projection.} Let $\mathcal W$ be defined as in \eqref{W-Q}. Apply
$O_K$ to $|0\rangle_w|\Sigma_{a,b}\rangle$, where $w$ is the workspace register.
By \eqref{sq-o-k}, this gives
\be \label{o-s} 
|\psi\rangle
=O_K(|0\rangle_w|\Sigma_{a,b}\rangle)
=\square|\Sigma_{a,b}\rangle .
\ee
Here $|0\rangle_w$ denotes the zero state of the workspace register. From~\eqref{proj} and \eqref{over} we therefore get, 
\begin{equation}\label{eq:projected-initial-state}
\Pi|\psi\rangle=\theta_{a,b}\,\square|\pi\rangle,
\end{equation}
where $\Pi$ is the orthogonal projector onto
$\ker(\mathcal W-I)$. 
Let $\eps \in (0,1)$ and set
\be \label{sdfs}
\eta=\frac{c\varepsilon}{4},
\qquad
\Delta_0=\sqrt{2\delta_0}.
\ee
Proposition~\ref{cor:SPUE-spectral}(ii) gives the phase-gap
bound $\Delta_0$ on $\mathcal L$, where on $\mathcal L^\perp$,
the only possible eigenvalues are $1$ and $-1$. Thus $\Delta_0$ also bounds every
nonzero eigenphase in absolute value on the full space.
Proposition~\ref{prop-poly} provides a polynomial $\Gamma$ and  a unitary $V$ such that $\|\Gamma(\mathcal W)-\Pi\|\leq\eta$, using
\be \label{poly-deg} 
d=O\!\left(\frac{1}{\sqrt{\delta_0}}
\log\frac{2}{c\varepsilon}\right)
\ee
controlled applications of $\mathcal W$.

Let $f$ denote the additional qubit of $V$ in \eqref{v-op}; outcome $f=0$
will signal success. Including preparation, define 
\begin{equation}\label{eq:full-filter-preparation}
B=(I_f\otimes O_K^*)\,V\,(I_f\otimes O_K)\,
  (I_f\otimes I_w\otimes A_{a,b}).
\end{equation}
All registers are initialised in their zero states. By \eqref{o-s},
\eqref{eq:projected-initial-state} and
$O_K^*\square|\pi\rangle=|0\rangle_w|\pi\rangle$, we obtain
\begin{equation}\label{eq:Phi-decomposition}
B|0\rangle
=|0\rangle_f|g\rangle+|1\rangle_f|h\rangle,
\qquad
|g\rangle=\theta_{a,b}|0\rangle_w|\pi\rangle+|e\rangle,
\end{equation}
where $|g\rangle$ and $|h\rangle$ are unnormalised states of
the workspace and path registers, and
\be \label{ggf} 
|e\rangle=O_K^*(\Gamma(\mathcal W)-\Pi)|\psi\rangle,
\qquad \|e\|\leq\eta.
\ee
Since \(|0\rangle_w|\pi\rangle\) is normalised and \(\theta_{a,b}\geq0\),
\[ \bigl\|\theta_{a,b}|0\rangle_w|\pi\rangle\bigr\|=\theta_{a,b}. \] 
Together with $\theta_{a,b}\geq c$, \eqref{sdfs}, \eqref{eq:Phi-decomposition} and \eqref{ggf} it follows that
\begin{equation}\label{eq:success-amplitude}
\bigl|\|g\|-\theta_{a,b}\bigr|\leq\eta,
\qquad
\|g\|\geq\theta_{a,b}-\eta\geq\frac{3c}{4}.
\end{equation}

\paragraph{Amplifying the flag.}
Apply Proposition~\ref{prop:amplitude-amplification} with
$\mathcal A=B$, marked projector
$\Pi_*=|0\rangle\langle0|_f\otimes I$, and amplitude lower
bound $3c/4$ for $c$ as in \eqref{c-dl-0}. The corresponding marking reflection
$S_*=I-2\Pi_*$ is a sign change on the flag-zero branch.
Let $Q_B$ be the Grover iterate \eqref{r0} for this preparation.
The proposition returns an integer $n=O(1/c)$ using
$O(1/c)$ controlled applications of $Q_B$.
Discard the estimation registers, prepare a fresh $B|0\rangle$,
apply $Q_B$ exactly $n$ times, and measure $f$.
By Remark~\ref{rem-grov}, the probability of obtaining $f=0$ is at least $32/(25\pi^2)$,
including both the estimation and final measurement outcomes.

For every returned $n$, equation
\eqref{eq:iteration-operator-Q-amplitude-amplification} shows
that the Grover iterations only change the scalar coefficients
of the two branches in \eqref{eq:Phi-decomposition}.
Conditioning on $f=0$ selects a scalar multiple of $|g\rangle$,
and normalisation removes the magnitude of this scalar, leaving only a global phase.
Consequently, whenever $f=0$, the remaining state is, up to
a global phase,
\begin{equation}\label{eq:rho-definition}
|\zeta\rangle=\frac{|g\rangle}{\|g\|}.
\end{equation}
This conditional state is the same even when the angle estimate
is inaccurate. If $\|g\|=1$, the input is already entirely
marked and the same conclusion holds.

\paragraph{Sampling error.}
Equations~\eqref{eq:Phi-decomposition}--\eqref{eq:rho-definition}
give
\begin{equation}\label{gg8}
\begin{aligned}
\bigl\||\zeta\rangle-|0\rangle_w|\pi\rangle\bigr\|
&=
\frac{\bigl\||g\rangle-\|g\|\,|0\rangle_w|\pi\rangle\bigr\|}
     {\|g\|}\\
&=
\frac{\bigl\|(\theta_{a,b}-\|g\|)|0\rangle_w|\pi\rangle
             +|e\rangle\bigr\|}
     {\|g\|}\\
&\leq
\frac{|\theta_{a,b}-\|g\||+\|e\|}{\|g\|}\\
&\leq\frac{2\eta}{\theta_{a,b}-\eta}
\leq\frac{c\varepsilon/2}{3c/4}
=\frac{2\varepsilon}{3}\leq\varepsilon, 
\end{aligned}
\end{equation}
Here the second equality uses \eqref{eq:Phi-decomposition},
and the first inequality uses the triangle inequality and
$\bigl\||0\rangle_w|\pi\rangle\bigr\|=1$. 
We now translate the state approximation in \eqref{gg8} into a bound
on the distribution of the sampled path. Conditional on $f=0$, the
workspace and path registers are in the normalised state
\[
|\zeta\rangle
=\sum_{z,x}\alpha_{z,x}|z\rangle_w|x\rangle,
\]
where $z$ labels a workspace basis state and
$x\in\Omega_{a,b}$ labels a complete bridge path. By the Born rule,
measuring both registers would give the pair $(z,x)$ with probability
$|\alpha_{z,x}|^2$. We retain only the path $x$, so its probability is
the sum over all workspace outcomes:
\[
\rho(x)=\sum_z|\alpha_{z,x}|^2.
\]
Discarding the workspace and measuring only the path register gives
this same distribution.

For comparison, the ideal state is
$|0\rangle_w|\pi\rangle
=\sum_x\sqrt{\pi(x)}\,|0\rangle_w|x\rangle$.
Its workspace is always zero, and measuring its path register gives
exactly the desired distribution $\pi$. Its coefficients in the same
basis are therefore
\[
\beta_{z,x}=
\begin{cases}
\sqrt{\pi(x)},&z=0,\\
0,&z\ne0,
\end{cases}
\qquad
\sum_z|\beta_{z,x}|^2=\pi(x).
\]
Using these expressions for $\rho(x)$ and $\pi(x)$, the definition
of total variation distance and the triangle inequality give
\[
\begin{aligned}
d_{\mathrm{TV}}(\rho,\pi)
&=\frac12\sum_x
\left|\sum_z\bigl(|\alpha_{z,x}|^2-|\beta_{z,x}|^2\bigr)\right|\\
&\leq\frac12\sum_{z,x}
\bigl||\alpha_{z,x}|^2-|\beta_{z,x}|^2\bigr|.
\end{aligned}
\]
Using $\bigl||a|^2-|b|^2\bigr|\leq|a-b|(|a|+|b|)$
for $a,b\in\mathbb C$, and then applying Cauchy--Schwarz, gives
\[
d_{\mathrm{TV}}(\rho,\pi)
\leq\frac12
\left(\sum_{z,x}|\alpha_{z,x}-\beta_{z,x}|^2\right)^{1/2}
\left(\sum_{z,x}(|\alpha_{z,x}|+|\beta_{z,x}|)^2\right)^{1/2}.
\]
Both states are normalised, so
$\sum_{z,x}|\alpha_{z,x}|^2
=\sum_{z,x}|\beta_{z,x}|^2=1$.
The second factor is consequently at most $2$, since
\[
\sum_{z,x}(|\alpha_{z,x}|+|\beta_{z,x}|)^2
\leq2\sum_{z,x}\bigl(|\alpha_{z,x}|^2+|\beta_{z,x}|^2\bigr)=4.
\]
The first factor is exactly the norm of the difference between the
two state vectors. Hence \eqref{gg8} gives
\[
d_{\mathrm{TV}}(\rho,\pi)
\leq\left(\sum_{z,x}|\alpha_{z,x}-\beta_{z,x}|^2\right)^{1/2}
=\bigl\||\zeta\rangle-|0\rangle_w|\pi\rangle\bigr\|
\leq\varepsilon.
\]
Since $\pi=\mathbb P_{a,b}$, the output path distribution is within
$\varepsilon$ of the desired Markov bridge law in total variation.
This guarantee is conditional on obtaining the success flag $f=0$.

\paragraph{Resource accounting.}
By \eqref{w-o-op}, each controlled walk call uses a constant
number of the primitives $O_K,O_K^*,S$.
Thus each use of $B$ or $B^*$ in
\eqref{eq:full-filter-preparation} uses $O(d)$ walk-primitive calls,
where $d$ is given by \eqref{poly-deg}, and one call to
$A_{a,b}$ or $A_{a,b}^*$, respectively. 
The estimation stage uses $O(1/c)$ controlled Grover calls,
and the final amplification uses $n=O(1/c)$ Grover calls, for $c$ in \eqref{c-dl-0}.
As in Remark \ref{rem-grov}, each such call
uses one $B$ and one $B^*$, together with the zero-state and
flag reflections. Including the preparations in both stages,
the total counts are
\begin{equation}\label{eq:generic-quantum-sampling-cost}
\begin{aligned}
Q_{\mathrm{walk}}
&=O\!\left(
\frac{1}{c\sqrt{\delta_0}}
\log\frac{2}{c\varepsilon}\right),\\
Q_{\mathrm{prep}}&=O(1/c),
\end{aligned}
\end{equation}
where $Q_{\mathrm{prep}}$ counts calls to $A_{a,b}$ and
$A_{a,b}^*$.
Substituting the values of $c$ and $\delta_0$ in \eqref{c-dl-0} gives
\[
Q_{\mathrm{walk}}
=O\!\left(
\frac{TN^2}{\omega_-^{5/2}}
\log\frac{N}{\omega_-\varepsilon}\right).
\]
Finally, the construction following \eqref{sigma} uses $T-1$
calls to $U_P$ for $A_{a,b}$, and $T-1$ calls to $U_P^*$ for
its inverse. Hence the total number of calls to $U_P,U_P^*$ is
\[
O\!\left(\frac{T-1}{c}\right)
=O\!\left(T\sqrt{\frac{N}{\omega_-}}\right).
\]
\end{proof}

\section{Proofs of Section \ref{sec:quantum-sinkhorn-alg}} \label{sec-pf-stat}
\begin{proof}[Proof of Lemma \ref{lem:static-scaling-bound}] 
We first argue that a minimiser exists. Positivity of the reference matrix and
target marginals allows us to apply
Proposition~\ref{prop:static-potentials}. The logarithms of its
positive scaling factors give vectors $(x^*,y^*)$ for which the
row and column sums are the prescribed marginals. Thus, by
\eqref{eq:schrodinger-equation}, the gradient of the convex function $f$ (see \eqref{eq:static-convex-potential}) at $(x^*,y^*)$ vanishes, so these vectors
minimise $f$.

We now use the freedom to shift the two potentials in opposite
directions. The matrix $M(x,y)$ in \eqref{eq:scaled-endpoint-matrix} is invariant under the transformation
$(x,y)\mapsto(x+s\mathbf{1},y-s\mathbf{1})$ for every
$s\in\mathbb R$. Since both marginals have mass one, this
transformation also leaves $f$ unchanged. We may therefore
choose a minimiser satisfying 
\be \label{maxd}
\max_{i\in\chi}x_i^*=0.
\ee  
From \eqref{eq:scaled-endpoint-matrix} and \eqref{lg1},
\[
e^{x_i^*}\sum_j\mathbb{P}_{0,T}(i,j)e^{y_j^*}=\mu_0(i), \quad \textrm{for all } i\in \chi.
\]
Comparing two row equations and using the assumption 
$p_{\min}\le\mathbb P_{0,T}(i,j)\le1$ and
$\mu_{0,\min}\le\mu_0(i)\le1$ gives
\[
e^{x_i^*-x_k^*}
=
\frac{\mu_0(i)}{\mu_0(k)}
\frac{\sum_j\mathbb{P}_{0,T}(k,j)e^{y_j^*}}
     {\sum_j\mathbb{P}_{0,T}(i,j)e^{y_j^*}}
\le \frac{1}{\mu_{0,\min}p_{\min}}, \quad \textrm{for all } i,k\in\chi.
\]
Taking logarithms and maximising over $i,k$ yields
\be  \label{ii0}
\max_{i\in\chi}x_i^*-\min_{i\in\chi}x_i^*
\le \log\!\left(\frac{1}{\mu_{0,\min}p_{\min}}\right).
\ee
The column constraints in \eqref{lg1} similarly give
\be \label{ii1}
e^{y_j^*}
=\frac{\mu_T(j)}
{\sum_{i\in\chi}\P_{0,T}(i,j)e^{x_i^*}}.
\ee
By \eqref{maxd}, some $x_{i_0}^*=0$ and every $x_i^*\le0$.
The denominator in \eqref{ii1} is therefore between $p_{\min}$
and $N$. Consequently,
\[
\log\mu_{T,\min}-\log N
\le y_j^*\le\log(1/p_{\min}).
\]
Combining this bound with \eqref{maxd} and \eqref{ii0} yields
\[
\|(x^*,y^*)\|_\infty
\le\max\left\{
\log\frac1{\mu_{0,\min}p_{\min}},\,
\log\frac N{\mu_{T,\min}},\,
\log\frac1{p_{\min}}
\right\}
\le B_0,
\]
which proves \eqref{eq:static-scaling-bound}.
\end{proof}
 
\begin{proof}[Proof of Theorem~\ref{th:quantum-static-bridge}]
We first check the hypotheses of the matrix-scaling algorithm (Algorithm~\ref{alg:quantum-static-bridge}), namely, that all parts of Assumption~\ref{assump-box} are satisfied. Recall Assumptions~\ref{assump-box}(i) and (iv) are part of the hypothesis of the theorem. Proposition~\ref{prop:static-potentials} provides an exact positive
scaling with the prescribed marginals, so
Assumption~\ref{assump-box}(ii) holds. By
Lemma~\ref{lem:static-scaling-bound}, an exact minimiser of $f$
has norm at most $B_0\le B$, which verifies
Assumption~\ref{assump-box}(iii). 

Let $\eps \in (0,1)$ and set $\eta=\eps/\sqrt3$ as in
Algorithm~\ref{alg:quantum-static-bridge}. By
\cite[Theorem~3.13]{Improved-quantum-lower-and-upper-bounds-for-matrix-scaling},
the algorithm returns vectors $(\bar x,\bar y)$ satisfying
\begin{equation}\label{eq:static-objective-error}
f(\bar x,\bar y)-f(x^*,y^*)\le6\eta^2
\end{equation}
with probability at least $2/3$. Since the reference matrix has
$N^2$ positive entries, the running time is
$\widetilde O(B^2N^{3/2}/\eta^2)$.
 
It remains to convert the optimisation error into an error for the
joint endpoint distribution. Recall from
\eqref{eq:approx-static-bridge} that
$\overline{\Q}_{0,T}=\overline M/L$, where $L>0$ is the total
mass of $\overline M$. By \eqref{eq:static-bridge-matrix-scaling},
$\Q_{0,T}^*=M(x^*,y^*)$ has total mass one. Taking the logarithm
of the ratio of the two couplings and using
the prescribed row and column sums of $\Q_{0,T}^*$ gives
\begin{align*}
\operatorname{KL}(\Q_{0,T}^*\,\|\,\overline{\Q}_{0,T})
&\quad=
\log L
+\sum_{i,j\in\chi}\Q_{0,T}^*(i,j)
\bigl(x_i^*+y_j^*-\bar x_i-\bar y_j\bigr)\\
&\quad=
\log L
+\langle\mu_0,x^*-\bar x\rangle
+\langle\mu_T,y^*-\bar y\rangle.
\end{align*}
Using \eqref{eq:static-convex-potential} and the total masses $L$ and $1$ of $\overline M$ and $\Q_{0,T}^*$, respectively, gives 
\[
f(\bar x,\bar y)-f(x^*,y^*)
=L-1+\langle\mu_0,x^*-\bar x\rangle
+\langle\mu_T,y^*-\bar y\rangle.
\]
Combining these identities yields
\be \label{eq:static-objective-KL-identity}
f(\bar x,\bar y)-f(x^*,y^*)
=L-1-\log L
+\operatorname{KL}(\Q_{0,T}^*\,\|\,\overline{\Q}_{0,T}).
\ee
Since $z-1-\log z\ge0$ for every $z>0$, 
\eqref{eq:static-objective-error} and
\eqref{eq:static-objective-KL-identity} imply
\[
\operatorname{KL}(\Q_{0,T}^*\,\|\,\overline{\Q}_{0,T})
\le6\eta^2.
\]
Pinsker's inequality therefore yields
\[
d_{TV}(\Q_{0,T}^*,\overline{\Q}_{0,T})
\le
\sqrt{\frac12
\operatorname{KL}(\Q_{0,T}^*\,\|\,\overline{\Q}_{0,T})}
\le\sqrt{3}\eta
 = \eps.
\]
The event \eqref{eq:static-objective-error} has probability at least
$2/3$, so this proves the stated probability and accuracy guarantees.
Substituting $\eta=\eps/\sqrt3$ gives the claimed running time.
\end{proof}
 
\section{Proofs of Section \ref{sec:extension-with-additive-cost}}  \label{sec-pr-additive} 
\begin{proof}[Proof of Proposition~\ref{prop-costs}]
Since $g$ is finite-valued on a finite state space,
\[
0<e^{-T\max_i g(i)}\le Z_T\le e^{-T\min_i g(i)}<\infty.
\]
Moreover, $\widehat{\mathbb V}$ and $\mathbb P$ have the same
support. If $\mathbb Q\not\ll\mathbb P$, both relative entropies
in \eqref{eq:cost-entropy-identity} are infinite, while the expected
cost is finite; the identity therefore holds in the extended-real
sense. For $\mathbb Q\ll\mathbb P$, substituting
$\widehat{\mathbb V}(x)=Z_T^{-1}\mathbb P(x)
\exp(-\sum_{t=0}^{T-1}g(x_t))$ into the definition of relative
entropy and summing over paths in the support of $\mathbb P$ gives 
\[
\begin{aligned}
\operatorname{KL}(\mathbb Q\Vert\widehat{\mathbb V})
&=\sum_x\mathbb Q(x)
\left(\log\frac{\mathbb Q(x)}{\mathbb P(x)}
+\sum_{t=0}^{T-1}g(x_t)+\log Z_T\right)\\
&=\operatorname{KL}(\mathbb Q\Vert\mathbb P)
+\mathbb E_{\mathbb Q}\left[\sum_{t=0}^{T-1}g(X_t)\right]
+\log Z_T.
\end{aligned}
\]
Rearranging proves \eqref{eq:cost-entropy-identity}.
The objective in \eqref{cost-add} and $\operatorname{KL}(\mathbb Q\Vert\widehat{\mathbb V})$ differ by the constant $\log Z_T$,
which is independent of $\mathbb Q$. The finite-feasibility
assumption ensures existence of a finite minimiser, and strict
convexity gives uniqueness. Thus their common minimiser is the
Schr\"odinger bridge with reference law $\widehat{\mathbb V}$.
\end{proof}

\begin{proof}[Proof of Proposition~\ref{prop-rev-g}]
Let $\pi_g=\widehat{\mathbb V}_{a,b}$. When all coordinates
except $x_t$ are fixed, the only varying factors in its path
weight are
$P(x_{t-1},x_t)P(x_t,x_{t+1})e^{-g(x_t)}$.
Normalising these factors gives
\eqref{eq:cost-gibbs-conditional}. Thus $K^g$ selects an interior
coordinate uniformly and resamples it from its conditional law.
Strict positivity of $P$ and finiteness of $g$ make these
conditional probabilities positive. As in the proof of
Proposition~\ref{prop-rev}, successive single-coordinate updates
connect any two paths, and resampling the current value leaves
the path unchanged with positive probability. This proves
irreducibility, aperiodicity, and the stated support of $K^g$.  

Define the conditional-expectation operator $E_t$ acting on
$f\in L^2(\pi_g)$ as
\[
(E_tf)(x)
=\mathbb E_{\pi_g}[f(X)\mid X_{-t}=x_{-t}]
=\sum_{s\in\chi}q^g(s,x_{t-1},x_{t+1})
  f(x^{t\leftarrow s}),
\]
where $x^{t\leftarrow s}$ denotes the path obtained from $x$
by replacing its $t$-th coordinate with $s$.
Since $K^g$ selects the coordinate to update uniformly,
\[
(K^gf)(x)
=\frac1{T-1}\sum_{t=1}^{T-1}(E_tf)(x),
\]
Each $E_t$ is an orthogonal projection. Hence $K^g$ is self-adjoint, which proves reversibility, and
\[
\langle f,K^g f\rangle_{\pi_g}
=\frac1{T-1}\sum_{t=1}^{T-1}
\|E_t f\|_{L^2(\pi_g)}^2\ge0.
\]
Hence the spectrum of $K^g$ lies in $[0,1]$.
\end{proof}

\begin{proof}[Proof of Theorem~\ref{cor:upper-bound-spectral-gap-g}]
We adapt the canonical-path argument from
Section~\ref{sec-pf-gibbs}. The key step is to cancel the cost
factors outside the coordinate being updated.

Recall that the matrix $W$ was defined in \eqref{W-def}. Define $w_i=e^{-g(i)}$ and use the transfer matrix
$R(i,j)=P(i,j)w_j$, which attaches the cost to the arrival state.
For a path with endpoints $a,b$,
\[
\prod_{t=0}^{T-1}R(x_t,x_{t+1})
=\frac{w_b}{w_a}\prod_{t=0}^{T-1}W(x_t,x_{t+1}).
\]
The factor $w_b/w_a$ is constant for fixed endpoints, so both
products give the same normalised bridge law
$\pi_g=\widehat{\mathbb V}_{a,b}$.

Use the paths \eqref{lead-p}--\eqref{gam}, whose lengths are
at most $T-1$. Define $\rho_g(x,y)$ and $\rho_g$ as in
\eqref{rhoxy}--\eqref{rho-def}, replacing $(\pi,K)$ by $(\pi_g,K^g)$.
Fix an off-diagonal edge $(x,y)$ that updates coordinate $t$
from $i=x_t$ to $j=y_t$, and write $\ell=x_{t-1}$ and $r=x_{t+1}$.
Define
\[
F_v=(R^t)(a,v),\qquad
H_v=(R^{T-t})(v,b),\qquad
Z_{a,b}=(R^T)(a,b)=\sum_v F_vH_v.
\]
The same prefix and suffix summation as in
\eqref{i1}--\eqref{f2}, now with $R$ in place of $P$, gives
\be \label{rho-g} 
\rho_g(x,y)
\le\frac{(T-1)^2}{q^g(j,\ell,r)}
\frac{R(\ell,j)}{R(\ell,i)}
\frac{F_iH_j}{Z_{a,b}}.
\ee
Expand $F_z$ over its last transition and $H_z$ over its first
transition. The weighted-ratio inequality used in \eqref{f4}
then gives, for every $z\in\chi$,
\be \label{eq:weighted-prefix-suffix} 
\frac{F_z}{F_i}\ge\frac{w_z}{w_i}
\min_u\frac{P(u,z)}{P(u,i)},\qquad
\frac{H_z}{H_j}\ge\min_v\frac{P(z,v)}{P(j,v)}.
\ee
These inequalities remain valid at $t=1$ and $t=T-1$, using
$R^0=I$. Set
\be \label{hh1}
\kappa_z=\min_{u,v}
\frac{P(u,z)P(z,v)}{P(u,i)P(j,v)},
\qquad C_g(\ell,r)=\sum_z P(\ell,z)P(z,r)w_z.
\ee
Multiplying the bounds in \eqref{eq:weighted-prefix-suffix} gives
\[
\frac{F_zH_z}{F_iH_j}\geq\frac{w_z}{w_i}\kappa_z.
\]
Summing this inequality over $z\in\chi$ and using
$Z_{a,b}=\sum_{z\in\chi}F_zH_z$, we obtain
\[
\frac{Z_{a,b}}{F_iH_j}
\geq\sum_{z\in\chi}\frac{w_z}{w_i}\kappa_z.
\]
By \eqref{eq:cost-gibbs-conditional} and the definition of $R$,
\[
\begin{aligned}
\frac1{q^g(j,\ell,r)}\frac{R(\ell,j)}{R(\ell,i)}
&=
\frac{C_g(\ell,r)}{P(\ell,j)P(j,r)w_j}
\frac{P(\ell,j)w_j}{P(\ell,i)w_i}\\
&=\frac{C_g(\ell,r)}{P(\ell,i)P(j,r)w_i}.
\end{aligned}
\]
Substituting these expressions into \eqref{rho-g} yields
\begin{equation}\label{eq:cost-congestion-bound}
\begin{aligned}
\rho_g(x,y)
&\leq (T-1)^2
\frac{C_g(\ell,r)}{P(\ell,i)P(j,r)w_i}
\frac{w_i}{\sum_{z\in\chi}w_z\kappa_z}\\
&=(T-1)^2
\frac{C_g(\ell,r)}
{P(\ell,i)P(j,r)\sum_{z\in\chi}w_z\kappa_z}.
\end{aligned}
\end{equation}
The positivity bound on $P$ in \eqref{om} implies
$\kappa_z\ge\omega_-^2/N^2$ (see \eqref{hh1}) and
$P(\ell,i)P(j,r)\ge\omega_-^2/N^2$.
Moreover, since $\sum_zP(\ell,z)P(z,r)\le1$ it follows from \eqref{del-g} and \eqref{hh1} that 
\be \label{jkl} 
C_g(\ell,r)\le e^{-\underline g},\qquad
\sum_z w_z\kappa_z
\ge e^{-\overline g}\frac{\omega_-^2}{N}.
\ee
Recall that $e^{-\underline g}/e^{-\overline g}=e^{\Delta_g}$. From \eqref{eq:cost-congestion-bound}, \eqref{jkl}, and \eqref{rho-def} we get 
\[
\rho_g\le \frac{T^2N^3}{\omega_-^4}e^{\Delta_g}.
\]
From Proposition \ref{prop:spectral-gap-notabsolute-to-absolute} we have $(1-\lambda_2(K^g))^{-1} \leq \rho_g$. Proposition~\ref{prop-rev-g} shows that the eigenvalues of $K^g$ are nonnegative, so by \eqref{sp-rev} the spectral gap of $K^g$ denoted by $\delta_g$ satisfies $\delta_g=1-\lambda_2(K^g)$. This proves \eqref{eq:cost-gibbs-gap}.
\end{proof}

\begin{proof}[Proof of Theorem~\ref{th:complexity-gibbs-sampler-g}]
Let $\pi_g=\widehat{\mathbb V}_{a,b}$ and define the row-preparation
isometry
\[
J_g|x\rangle=\sum_y\sqrt{K^g(x,y)}\,|x,y\rangle.
\]
Apply $J_g$ to the initial state $|r_g\rangle$ in \eqref{gfd}.
The resulting state belongs to $\operatorname{ran}J_g$, and hence to
\[
\mathcal L_g=\operatorname{span}
\{\operatorname{ran}J_g,\,S\operatorname{ran}J_g\}.
\]
The swap $S$ exchanges the two subspaces in this span, while
the reflection $2J_gJ_g^*-I$ also preserves their span.
Consequently, $\mathcal L_g$ is invariant under
$\mathcal W_g=(2J_gJ_g^*-I)S$.

By Proposition~\ref{prop-rev-g}, $K^g$ is irreducible,
reversible with respect to $\pi_g$, and positive semidefinite.
Thus the argument of Proposition~\ref{cor:SPUE-spectral}
in Section~\ref{sec-q-gibbs} applies with the substitutions
\[
(K,\pi,\square,\mathcal W,\mathcal L,\delta)
\longmapsto
(K^g,\pi_g,J_g,\mathcal W_g,\mathcal L_g,\delta_g).
\]
Consequently,
\[
\ker(\mathcal W_g-I)\cap\mathcal L_g
=\operatorname{span}\{J_g|\pi_g\rangle\},
\]
and every nonzero eigenphase of
$\mathcal W_g|_{\mathcal L_g}$ has magnitude at least
$\arccos(1-\delta_g)\geq\sqrt{2\delta_g}$.

The isometry preserves the initial overlap in \eqref{g-over}, 
\[
\bigl|\langle J_g\pi_g\mid J_gr_g\rangle\bigr|
=\bigl|\langle\pi_g\mid r_g\rangle\bigr|\ge c_g.
\]

We now repeat the construction in Section~\ref{pf-gibbs}, using
$A_g$ and $|r_g\rangle$ in place of $A_{a,b}$ and
$|\Sigma_{a,b}\rangle$, and making the substitutions
\[
(\pi,K,\square,\mathcal W,c,\delta_0)
\longmapsto
(\pi_g,K^g,J_g,\mathcal W_g,c_g,\delta_{g,0}).
\]

Under these substitutions, the overlap becomes
$\theta_g:=\langle\pi_g|r_g\rangle\geq c_g$, since both
states have nonnegative amplitudes. The parameter choices
in \eqref{sdfs} become
\[
\eta=\frac{c_g\varepsilon}{4},
\qquad
\Delta_0=\sqrt{2\delta_{g,0}}.
\]
As in Section~\ref{sec-q-gibbs},
$\mathcal W_g=-S$ on $\mathcal L_g^\perp$, so $\Delta_0$
also bounds the nonzero eigenphases on the full space.
The bound \eqref{eq:success-amplitude}, with $c$ replaced
by $c_g$, gives a success amplitude of at least $3c_g/4$.
The amplification argument in Section~\ref{pf-gibbs}
therefore gives success probability at least $32/(25\pi^2)$.
Using the same decoding step as in
\eqref{eq:full-filter-preparation}, followed by discarding
the workspace and measuring the path register,
\eqref{gg8} and the subsequent measurement estimate give
\[
d_{\mathrm{TV}}(\rho,\pi_g)\leq\varepsilon
\]
conditional on success.

Finally, substituting $c_g$ and $\delta_{g,0}$ for
$c$ and $\delta_0$ in \eqref{eq:generic-quantum-sampling-cost}
gives
\[
O\!\left(
\frac{1}{c_g\sqrt{\delta_{g,0}}}
\log\frac{2}{c_g\varepsilon}
\right)
\]
calls to the walk primitives for $K^g$, and $O(1/c_g)$
calls to $A_g$ and its inverse.

Taking
\[
\delta_{g,0}
=\frac{\omega_-^4e^{-\Delta_g}}{T^2N^3}
\]
from Theorem~\ref{cor:upper-bound-spectral-gap-g}
in the preceding query bound yields
\[
O\!\left(
\frac{TN^{3/2}e^{\Delta_g/2}}{c_g\omega_-^2}
\log\frac{2}{c_g\varepsilon}
\right),
\]
which is \eqref{comp-gibbs-g}. The preparation cost remains
$O(1/c_g)$ calls to $A_g$ and its inverse.

\end{proof}

\begin{proof}[Bounds used in Remark~\ref{rem:cost-classical-comparison}]
Write $\pi_g=\widehat{\mathbb V}_{a,b}$.
By the bridge formula \eqref{m-br}, the positivity bound
\eqref{om}, and $P^T(a,b)\leq1$, every $x\in\Omega_{a,b}$ satisfies
\[
\mathbb P_{a,b}(x)
=\frac{\prod_{t=0}^{T-1}P(x_t,x_{t+1})}{P^T(a,b)}
\geq\left(\frac{\omega_-}{N}\right)^T.
\]
Conditioning the weighted reference law
\eqref{eq:definition-of-R}--\eqref{v-hat} on the endpoints
cancels the fixed factor $e^{-g(a)}$. Hence
\[
\pi_g(x)
=\frac{\mathbb P_{a,b}(x)e^{-\sum_{t=1}^{T-1}g(x_t)}}
{\mathbb E_{\mathbb P_{a,b}}
 [e^{-\sum_{t=1}^{T-1}g(X_t)}]}
\geq \mathbb P_{a,b}(x)e^{-(T-1)\Delta_g},
\]
where the inequality uses the cost bounds in \eqref{del-g}.
Consequently,
\[
\log\frac1{\min_x\pi_g(x)}
\leq T\log\frac N{\omega_-}+(T-1)\Delta_g.
\]
Proposition~\ref{prop-rev-g} ensures that $K^g$ is reversible
and ergodic. Applying Theorem~\ref{thm-mix} with this estimate
and the gap bound \eqref{eq:cost-gibbs-gap} proves the
fixed-path update bound in
Remark~\ref{rem:cost-classical-comparison}.
 \end{proof} 
 
\section*{Acknowledgements}
The authors would like to thank Yifan Jiang for helpful discussions that significantly improved the quality of this paper.

\bibliographystyle{abbrvnat}
\bibliography{references}

\appendix

\section{Numerical Implementation of Endpoint Sampling}
\label{app:endpoint-precision}

This appendix justifies the numerical implementation
of the endpoint sampler in
Section~\ref{rem:static-output}, under the assumptions
of Theorem~\ref{th:quantum-static-bridge}.
Fix the scaling vectors $(\bar x,\bar y)$ returned by
Algorithm~\ref{alg:quantum-static-bridge}.
They specify the approximate coupling
$\overline{\Q}_{0,T}$ in
\eqref{eq:approx-static-bridge}.

The ideal sampling procedure draws $X_0\sim\mu_0$
and, conditional on $X_0=i$, draws $X_T$ from
\begin{equation}\label{eq:appendix-conditional-row}
    q_i(j)
    =\frac{\P_{0,T}(i,j)e^{\bar y_j}}
           {\sum_k\P_{0,T}(i,k)e^{\bar y_k}}.
\end{equation}
Its joint law is $\widetilde{\Q}_{0,T}$, defined in
\eqref{eq:sampled-static-bridge}. We show how to
implement this procedure with controlled numerical
error under Assumption~\ref{assump-box}(iv).

\paragraph{Error introduced by normalisation.}
Let $v_j\ge0$ have total mass $V=\sum_jv_j>0$, and
let $\widehat v_j\ge0$ satisfy
\[
    E:=\sum_j|\widehat v_j-v_j|<V.
\]
Then $\widehat V:=\sum_j\widehat v_j>0$, and
\[
\begin{aligned}
    \sum_j\left|
        \frac{v_j}{V}-\frac{\widehat v_j}{\widehat V}
    \right|
    &\le\frac{E}{V}
       +\frac{|\widehat V-V|}{V}\\
    &\le\frac{2E}{V}.
\end{aligned}
\]
Consequently,
\begin{equation}\label{eq:normalization-precision-bound}
    d_{\mathrm{TV}}\!\left(
        (v_j/V)_j,(\widehat v_j/\widehat V)_j
    \right)\le\frac{E}{V}.
\end{equation}
Thus, controlling the total absolute error in the
weights controls the error after normalisation.

\paragraph{Approximating a conditional row.}
Choose $j_*\in\argmax_j\bar y_j$ and define
\begin{equation}\label{eq:appendix-rescaled-weights}
    v_j=\P_{0,T}(i,j)e^{\bar y_j-\bar y_{j_*}}.
\end{equation}
Normalising these weights gives exactly
\eqref{eq:appendix-conditional-row}.
The positivity assumptions in
Theorem~\ref{th:quantum-static-bridge} ensure that
\[
    0<v_j\le1,
    \qquad
    V:=\sum_jv_j
    \ge p_i^\circ:=\P_{0,T}(i,j_*)>0.
\]
The lower bound $p_i^\circ$ requires only one
entry query in the selected row.

Fix $0<\xi_1<1$. Compute nonnegative dyadic
approximations satisfying
\begin{equation}\label{eq:appendix-weight-tolerance}
    |\widehat v_j-v_j|
    \le\frac{\xi_1p_i^\circ}{2N}
    \qquad\text{for every }j.
\end{equation}
Then $E\le\xi_1p_i^\circ/2<V$, and
\eqref{eq:normalization-precision-bound} gives
\[
    d_{\mathrm{TV}}\!\left(
        q_i,(\widehat v_j/\widehat V)_j
    \right)\le\frac{\xi_1}{2}.
\]
The remaining $\xi_1/2$ can accommodate further
normalisation or sampling errors. Alternatively,
the normalised dyadic weights can be sampled exactly
by the procedure below. Weights below the tolerance
in \eqref{eq:appendix-weight-tolerance} may be
rounded to zero.

Set $p_{\min}=\min_{i,j}\P_{0,T}(i,j)$.
Since $p_i^\circ\ge p_{\min}$, the approximations
in \eqref{eq:appendix-weight-tolerance} can be
represented using
\[
    O\!\left(\log\frac{N}{p_{\min}\xi_1}\right)
\]
fractional bits. Assumption~\ref{assump-box}(i)
bounds $p_{\min}$ below by an inverse polynomial
in $N$, making this bit count logarithmic in $N$
and $1/\xi_1$. The value of $p_{\min}$ need not
be computed: the selected entry $p_i^\circ$
determines the tolerance.
The precision required to read the scaling vectors
and evaluate \eqref{eq:appendix-rescaled-weights}
is accounted for under
Assumption~\ref{assump-box}(iv).

\paragraph{Sampling the approximate distributions.}
For the first marginal, take nonnegative dyadic
approximations to $\mu_0(i)$ with absolute error
at most $\xi_0/(2N)$, where $0<\xi_0<1$.
Since $\sum_i\mu_0(i)=1$, equation
\eqref{eq:normalization-precision-bound} bounds
the total variation error after normalisation
by $\xi_0/2$.

To sample normalised dyadic weights exactly,
write them with a common denominator as $a_j/2^b$,
where the $a_j$ are nonnegative integers, and set
$A=\sum_j a_j>0$.
Draw an integer uniformly from $\{0,\ldots,A-1\}$
and locate it in the cumulative sums of the $a_j$.
For $A>1$, a uniform integer is obtained using
$\lceil\log_2 A\rceil$ unbiased random bits and
rejecting values at least $A$. Each attempt succeeds
with probability greater than $1/2$, so the expected
number of attempts is less than two.
The case $A=1$ is deterministic.

Finding $j_*$, evaluating
\eqref{eq:appendix-rescaled-weights}, and forming
the cumulative sums for the selected row and for
$\mu_0$ require $O(N)$ entry evaluations and
scalar operations. This establishes the additional
sampling cost stated in
Section~\ref{rem:static-output}, with precision
and random-bit costs accounted for under
Assumption~\ref{assump-box}(iv).

\paragraph{Total endpoint error.}
Let $\widehat\mu_0$ and $\widehat q_i$ be the
distributions actually sampled, satisfying
\[
    d_{\mathrm{TV}}(\widehat\mu_0,\mu_0)\le\xi_0,
    \qquad
    \sup_i d_{\mathrm{TV}}(\widehat q_i,q_i)\le\xi_1.
\]
The implemented endpoint law is
$\widehat{\Q}_{0,T}(i,j)=\widehat\mu_0(i)\widehat q_i(j)$.
Comparing it with
$\widetilde{\Q}_{0,T}$ in
\eqref{eq:sampled-static-bridge}, and inserting
the intermediate law $\mu_0(i)\widehat q_i(j)$,
gives
\[
\begin{aligned}
    d_{\mathrm{TV}}(
        \widehat{\Q}_{0,T},\widetilde{\Q}_{0,T})
    &\le d_{\mathrm{TV}}(\widehat\mu_0,\mu_0)
       +\sum_i\mu_0(i)
          d_{\mathrm{TV}}(\widehat q_i,q_i)\\
    &\le\xi_0+\xi_1.
\end{aligned}
\]
Let $\varepsilon_s$ denote the accuracy parameter
used in Theorem~\ref{th:quantum-static-bridge}.
For any scaling output satisfying
$d_{\mathrm{TV}}(\overline{\Q}_{0,T},\Q_{0,T}^*)
\le\varepsilon_s$, the triangle inequality and
\eqref{eq:sampled-static-bridge-TV} yield
\[
    d_{\mathrm{TV}}(\widehat{\Q}_{0,T},\Q_{0,T}^*)
    \le2\varepsilon_s+\xi_0+\xi_1.
\]
This proves
\eqref{eq:sampled-static-bridge-precision}.

On the stronger event established in
\eqref{eq:static-objective-error}--\eqref{eq:static-objective-KL-identity},
we have
$\operatorname{KL}(\Q_{0,T}^*\|\overline{\Q}_{0,T})
\le2\varepsilon_s^2$.
Then \eqref{eq:sampled-static-bridge-KL} and
Pinsker's inequality improve the endpoint bound to
\[
    d_{\mathrm{TV}}(\widehat{\Q}_{0,T},\Q_{0,T}^*)
    \le\varepsilon_s+\xi_0+\xi_1.
\]
Both guarantees apply to each fixed realisation
of the scaling vectors satisfying the corresponding
error bound.
\end{document}